\documentclass[journal,10pt]{IEEEtran}
\usepackage{cite}
\usepackage{amsmath,amssymb,amsfonts}
\usepackage{graphicx}
\usepackage{textcomp}
    \usepackage{graphicx}
\usepackage{amsthm}
\usepackage{color,soul}
\usepackage{xcolor}
\usepackage{lineno,hyperref}
\usepackage{epsfig}
\usepackage{subfigure}
\usepackage{amssymb}
\usepackage{algorithm}
\usepackage{algpseudocode}
\usepackage{siunitx}
\usepackage{dirtytalk}
\usepackage{amsmath}
\usepackage{latexsym}
\usepackage{xfrac}
\usepackage{textcomp}
\usepackage{amsfonts}
\hypersetup{
    colorlinks=true,
    linkcolor=red,
    citecolor=blue,
    }
\usepackage{scrextend}
\usepackage{amsmath}
\usepackage{blindtext}
\usepackage{pbox}
\usepackage{lineno,hyperref}
\usepackage{epsfig}
\usepackage{subfigure}
\usepackage{amssymb}
\usepackage{siunitx}
\usepackage{latexsym}
\usepackage{dirtytalk}
\usepackage{amsmath}
\usepackage{xfrac}
\usepackage{textcomp}
\usepackage{amsfonts}
\usepackage{float}
\usepackage{blindtext}
\usepackage{xcolor,colortbl}
\definecolor{Gray}{gray}{0.5}
\usepackage{longtable}
\usepackage{xcolor}
\usepackage{soul}
\modulolinenumbers[5]

\definecolor{Gray}{gray}{0.9}
\usepackage{xcolor}
\usepackage{lineno}
\newtheorem{rem}{Remark}

\newtheorem{thm}{Theorem}
\newtheorem{assump}{Assumption}

\begin{document}

\title{Cyber-Resilient Data-Driven Event-Triggered Secure Control for Autonomous Vehicles Under False Data Injection Attacks}

\author{Yashar Mousavi,~\IEEEmembership{Member,~IEEE,}
        Mahsa Tavasoli, Ibrahim Beklan Kucukdemiral,~\IEEEmembership{Senior Member,~IEEE,}
        Umit Cali,~\IEEEmembership{Senior Member,~IEEE,} 
        Abdolhossein Sarrafzadeh,
        Ali Karimoddini,~\IEEEmembership{Senior Member,~IEEE,}
        and~Afef~Fekih,~\IEEEmembership{Senior~Member,~IEEE}
        
\thanks{Y. Mousavi and I.B. Kucukdemiral are with the Department of Engineering, School of Science and Engineering, Glasgow Caledonian University, Glasgow G4 0BA, UK; e-mail: yashar.mousavi@gcu.ac.uk, ibrahim.kucukdemiral@gcu.ac.uk.}
\thanks{M. Tavasoli is with the Department of Applied Science \& Technology, North Carolina Agricultural and Technical State University, Greensboro, NC 27411, USA, email: mtavasoli@aggies.ncat.edu}
\thanks{A. Sarrafzadeh is with the Department of Electrical and Computer Engineering, Old Dominion University, Norfolk, VA 23529, USA, email: asarrafz@odu.edu}
\thanks{U. Cali is with the School of Physics, Engineering and Technology, University of York, YO10 5DD York, United Kingdom, and the Department of Electric Energy, Norwegian University of Science and Technology, Trondheim, Norway; email: umit.cali@ntnu.no.}
\thanks{A. Karimoddini is with the Department of Electrical and Computer Engineering, North Carolina Agricultural and Technical State University, Greensboro, NC 27411, USA, email: akarimod@ncat.edu.}
\thanks{A. Fekih is with the Electrical and Computer Engineering Department, University of Louisiana at Lafayette, P.O. Box 43890, Lafayette, LA 70504, USA, email: afef.fekih@louisiana.edu.}
}


\maketitle

\begin{center}
\textbf{Notice:}
This work has been submitted to the IEEE for possible publication.
Copyright may be transferred without notice, after which this version
may no longer be accessible.
\end{center}

\begin{abstract}
This paper proposes a cyber-resilient secure control framework for autonomous vehicles (AVs) subject to false data injection (FDI) threats as actuator attacks. The framework integrates data-driven modeling, event-triggered communication, and fractional-order sliding mode control (FSMC) to enhance the resilience against adversarial interventions. A dynamic model decomposition (DMD)-based methodology is employed to extract the lateral dynamics from real-world data, eliminating the reliance on conventional mechanistic modeling. To optimize communication efficiency, an event-triggered transmission scheme is designed to reduce the redundant transmissions while ensuring system stability. Furthermore, an extended state observer (ESO) is developed for real-time estimation and mitigation of actuator attack effects. Theoretical stability analysis, conducted using Lyapunov methods and linear matrix inequality (LMI) formulations, guarantees exponential error convergence. Extensive simulations validate the proposed event-triggered secure control framework, demonstrating substantial improvements in attack mitigation, communication efficiency, and lateral tracking performance. The results show that the framework effectively counteracts actuator attacks while optimizing communication-resource utilization, making it highly suitable for safety-critical AV applications.
\end{abstract}


\section{Introduction}
\label{sec:1}
\IEEEPARstart{A}UTONOMOUS vehicles (AVs) are transforming modern transportation by integrating advanced communication, sensing, and control technologies to enhance safety, efficiency, and reliability. These systems rely heavily on real-time decision-making to navigate complex environments, demanding robust control mechanisms for precise lateral regulation. Lateral control is particularly important in AVs as it directly influences lane-keeping, path-following, and overall stability. However, with the increasing reliance on networked control systems, AVs face several challenges, including modeling complexities, communication constraints, and cybersecurity threats, all of which necessitate secure and adaptive control strategies.

The significance of AVs extends beyond convenience; they have the potential to drastically reduce traffic-related fatalities, minimize congestion, and optimize energy consumption. Recent advancements in ADAS \cite{argui2024advancements} have paved the way for fully autonomous navigation, with major manufacturers heavily investing in self-driving technologies. Despite these developments, large-scale implementation remains challenging due to the interplay of vehicle dynamics, environmental factors, and cyber-physical security risks. The lateral dynamics of AVs have been extensively studied, with various control approaches proposed to enhance stability and tracking accuracy. Early research focused on conventional methods, such as LQR approaches and fuzzy logic control \cite{arab2020optimal}. While effective for structured environments, these techniques often rely on linearized models that fail to capture nonlinear behavior. To address these limitations, model predictive control \cite{feng2023distributed} and robust control frameworks, including $H_{\infty}$ control \cite{zhu2023survey} and sliding mode control (SMC) \cite{lee2024novel}, have been explored. Among these, SMC has gained significant attention due to its inherent robustness to uncertainties and external disturbances \cite{wang2022path}. However, traditional SMC methods often suffer from chattering effects \cite{mousavi2022sliding}. To mitigate this, higher-order SMCs \cite{sun2022nonsingular} and fractional-order SMCs (FSMCs) \cite{taghavifar2024nonsingleton} have been proposed. FSMC leverages fractional calculus to improve robustness, reduce chattering, and enhance dynamic response, motivating its adoption in this work.

Recent advances in data-driven techniques provide promising solutions for overcoming modeling challenges. Instead of relying on simplified mechanism-based models, data-driven approaches such as DMD \cite{perrusquia2024trajectory} and Koopman operator theory \cite{chen2024incorporating} capture vehicle dynamics directly from sensor data. Additionally, intelligent controllers leveraging deep learning \cite{kuutti2020survey}, reinforcement learning \cite{ma2024game}, and neural networks \cite{hedesh2025delay} have shown significant improvements. However, their reliance on extensive datasets and computationally intensive training limits real-time applicability in safety-critical systems, highlighting the need for computationally efficient control strategies that are robust against cyber-attacks and communication constraints.

Another major challenge is the constraint imposed by limited communication bandwidth. As AVs generate vast amounts of sensor and control data, efficient transmission mechanisms are necessary. Traditional time-driven communication schemes lead to excessive data exchange, overwhelming in-vehicle networks. Event-triggered control (ETC) strategies \cite{zhang2023event,wong2023novel} mitigate these effects by transmitting control updates only when predefined conditions are met, significantly reducing communication demands while maintaining stability. However, ensuring secure transmissions in the presence of cyber threats remains challenging. 

Cyber-attacks pose significant challenges in AV control due to the increasing prevalence of FDI attacks and actuator manipulations. In-vehicle networks, such as CAN, are inherently vulnerable due to their lack of built-in security mechanisms \cite{sargolzaei2019detection}. Attackers can exploit these weaknesses by injecting false signals, leading to incorrect control decisions and potential instability, or introducing artificial time delays that disrupt stability mechanisms \cite{oshnoei2023model}. Various attack detection methods, including statistical anomaly detection \cite{zhang2024anomaly}, Kalman filtering \cite{lam2024aoa}, observer-based approaches \cite{shu2024adaptive}, have been developed. However, most solutions focus solely on detection rather than secure control or do not emphasize real-time compensation, leaving AVs susceptible to disruptions even after attack detection. While traditional fault-tolerant control methods have been employed to mitigate actuator failures \cite{lu2021event,sun2023secure}, they remain ineffective against stealthy, model-aware cyber-attacks, necessitating observer-based strategies that can reconstruct true system states in the presence of deceptive attacks.

Given these limitations, this paper introduces a novel data-driven event-triggered secure fractional-order SMC (ETS-FSMC) framework to mitigate actuator attacks and FDI threats in AVs. Unlike existing works that focus on detection or passive mitigation, the proposed method actively neutralizes actuator attacks through an adaptive observer-based approach. While previous observer-based methods \cite{li2024event,shu2024adaptive} primarily estimate attack signatures without real-time compensation or employ static thresholds, our approach differs in three key aspects: (1) it integrates the extended state observer directly within the control loop for simultaneous estimation and compensation, enabling dynamic attack neutralization rather than just detection; (2) it combines fractional-order sliding control with observer-based estimation to provide robust tracking performance even under persistent attacks; and (3) it incorporates event-triggered communication within the secure control framework, addressing both security and efficiency concerns simultaneously—a combination not explored in existing literature. This work focuses on actuator-side FDI attacks that are stealthy, persistent, and bounded, with attackers assumed to have partial or full model knowledge. By integrating data-driven modeling, event-triggered transmission, and secure control mechanisms, this paper establishes a resilient control framework for AV lateral dynamics that explicitly accounts for stealthy, model-aware actuator attacks, ensuring robust estimation, real-time compensation, and resilient closed-loop stability in adversarial settings. The key contributions are:
\begin{itemize}
    \item Developing a data-driven modeling framework leveraging dynamic mode decomposition to extract critical lateral dynamics from real-time data, reducing reliance on conventional mechanistic models.
    \item Developing a secure event-triggered control scheme that optimally balances communication efficiency and system stability, ensuring robustness against transmission delays and network constraints.
    \item Designing a fractional-order sliding mode-based attack mitigation strategy integrating an extended state observer to actively counteract actuator attacks and false data injections, maintaining system stability even under adversarial conditions.
    \item Conducting a comprehensive stability analysis based on Lyapunov theory, providing theoretical guarantees for exponential error convergence and attack resilience.
\end{itemize}

The remainder of this paper is structured as follows. Section \ref{sec:2} formulates the problem and presents the system model. Section \ref{sec:3} introduces the proposed data-driven secure event-triggered control framework. Section \ref{sec:4} presents the main theoretical results. Section \ref{sec:5} presents simulation results and performance evaluations. Finally, Section \ref{sec:6} concludes the study.

\section{Problem Statement and Methodological Framework}
\label{sec:2}

\subsection{Mathematical Modeling of Autonomous Vehicle Lateral Dynamics}
\label{sec:2.1}
The control architecture for AVs fundamentally relies on precise management of lateral motion and yaw dynamics for effective path-following operations. To facilitate systematic analysis of these dynamics, a reduced-order representation based on a two-degree-of-freedom monorail vehicle model for path tracking analysis can be employed (see Fig. \ref{fig:Vhcl}). Supposing the constant-velocity operation with small front wheel angular deflections, the small angle approximation principle ($\cos(\theta^f) \approx 1$) can be applied, enabling to express the lateral dynamics as:
\begin{align}
\label{Eq__1}
m\dot{a}_y &= m(\dot{v}_y + \nu_x\dot{\phi}) = F^f_y\cos(\theta^f) + F^r_y = F^f_y + F^r_y, \nonumber\\
\mathcal I_z\ddot{\phi} &= L^fF^f_y\cos(\theta^f) - L^rF^r_y = L^fF^f_y - L^rF^r_y,
\end{align}
where $m$ denotes the total vehicle mass, $\theta^f$ represents the front wheel steering angle, $\phi$ indicates the vehicle heading angle, and $a_y$ characterizes the lateral acceleration. $\dot{\phi}$ and $\ddot{\phi}$ represent the yaw rate and yaw acceleration, respectively. The vehicle's kinematic state is further described by its longitudinal velocity $\nu_x$ and lateral velocity $\nu_y$, measured at the center of mass. The lateral forces acting on the front and rear tires are denoted by $ F^f_y$ and $F^r_y$, respectively. The rotational dynamics are governed by $\mathcal I_z$, the moment of inertia about the yaw axis, while $L^f$ and $L^r$ represent the respective distances from the center of mass to the front and rear axles. It is worth mentioning that the derivation assumes negligible longitudinal slip and uniform road conditions, ensuring the accuracy of the reduced-order model under small deflections.
\begin{figure}[!t]
\centering
\includegraphics[width=2.7 in]{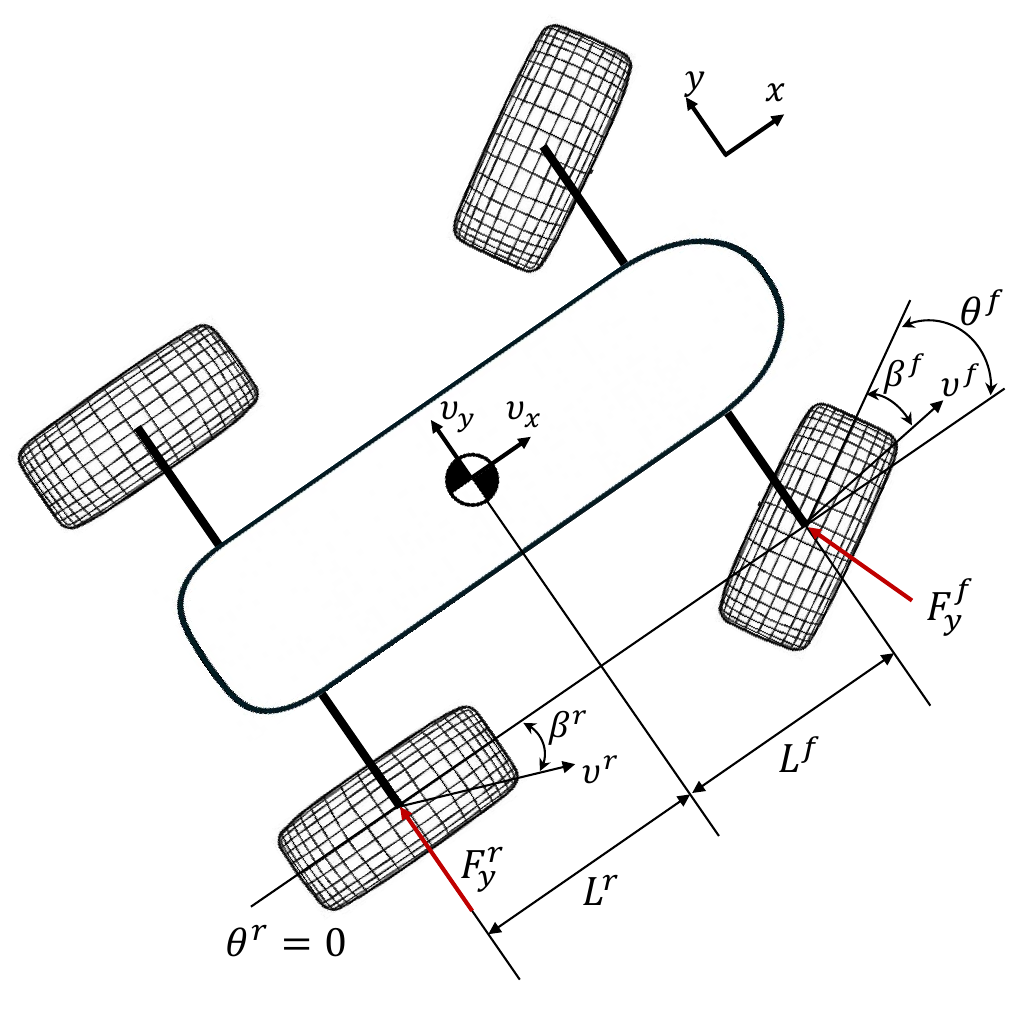}
\caption{Kinematic model of AV lateral dynamic.}
\label{fig:Vhcl}
\end{figure}

\begin{rem}
\label{rem:1}
The small angle approximation represents a mathematically justified simplification under typical operating conditions, where front wheel steering angles remain within a limited range $\theta^f<5^\circ$. This approximation preserves model fidelity while enabling analytical tractability. The application of small angle approximation allows for the systematic elimination of higher-order nonlinear terms, enabling focused analysis of the dominant factors governing vehicle lateral dynamics.
\end{rem}

Under the assumption of small slip angles, which is consistent with normal operating conditions, a linear tire force model can be developed as:
\begin{align}
\label{Eq__2}
F^f_y &= C^f_{\beta}\beta^f = C^f_{\beta}\left(\theta^f - \left(\nu_y+a\dot{\phi}\right)/\nu_x\right), \nonumber\\
F^r_y &= C^r_{\beta}\beta^r = C^r_{\beta}\left(L^r\dot{\phi}-\nu_y\right)/\nu_x,
\end{align}
where  $C^f_{\beta}$ and $C^r_{\beta}$ represent the cornering stiffness coefficients of the front and rear tires, respectively, and the angles $\beta^f$ and $\beta^r$ denote the side slip angles at the front and rear wheels.

\begin{rem}
\label{rem:2}
The linear tire model approximation is mathematically justified under small slip angle conditions, establishing a first-order relationship between tire forces and slip angles that captures the essential dynamics of the system. This approximation is applicable in normal driving scenarios with lateral acceleration below $0.5g$, beyond which nonlinear effects may dominate.
\end{rem}

Under the small-angle approximation established in Remark \ref{rem:1}, the nonlinear tire forces in \eqref{Eq__2} can be linearized. For small slip angles $\beta^f$ and $\beta^r$, the sine functions of these angles can be approximated by the angles themselves ($\sin \beta \approx \beta$), and the cosine terms approach unity ($\cos \beta \approx 1$). Applying these approximations to \eqref{Eq__2} and substituting the linearized tire forces into \eqref{Eq__1} yields the lateral dynamics expressed as:

{\small
\begin{align}
\label{Eq__3}
\dot{v}_y &= -\frac{C^f_{\beta}+C^r_{\beta}}{m\nu_x}\nu_y + \left(-\frac{L^fC^f_{\beta}+L^rC^r_{\beta}}{m\nu_x}-\nu_x\right)\dot{\phi} + \frac{C^f_{\beta}}{m}\theta^f, \\
\ddot{\phi} &= -\frac{L^fC^f_{\beta}-L^rC^r_{\beta}}{\mathcal I_z\nu_x}\nu_y + \left(-\frac{L^{f^2}C^f_{\beta}+L^{r^2}C^r_{\beta}}{\mathcal I_z\nu_x}\right)\dot{\phi} + \frac{L^fC^f_{\beta}}{\mathcal I_z}\theta^f.\nonumber
\end{align}}

\normalsize

In order to capture a precise path following control, two essential error metrics are defined as the lateral position error $e_d$ and the heading angle error $e_\phi$, quantifying the deviation of the vehicle's position and heading from the desired path, governed by:
\begin{align}
\label{Eq__4}
e_\phi &= \phi - \phi_{des} \nonumber\\
\dot{e}_d &= \nu_y + \nu_x(\phi - \phi_{des})
\end{align}
where $\phi_{des}$ represents the desired reference heading angle. 

Incorporating the error definitions \eqref{Eq__4} into the lateral dynamics model \eqref{Eq__3} yields the error dynamics model, which forms the basis for control design:

\small{
\begin{align}
\label{Eq__5}
\ddot{e}_d &= -\frac{C^f_{\beta}+C^r_{\beta}}{m\nu_x}e_d + \frac{C^f_{\beta}+C^r_{\beta}}{m}e_\phi + \left(-\frac{L^fC^f_{\beta}-L^rC^r_{\beta}}{m\nu_x}\right)\dot{e}_\phi  \nonumber\\ 
&\,\,\,\,\, + \frac{C^f_{\beta}}{m}\theta^f + \left(-\nu_x-\frac{L^fC^f_{\beta}-L^rC^r_{\beta}}{m\nu_x}\right)\dot{\phi}_{des},  \nonumber\\
\ddot{e}_\phi &= -\frac{L^fC^f_{\beta}-L^rC^r_{\beta}}{\mathcal I_z\nu_x}e_d + \left(\frac{L^fC^f_{\beta}-L^rC^r_{\beta}}{\mathcal I_z}\right)e_\phi \nonumber \\
&\,\,\,\,\,+ \left(-\frac{L^{f^2}C^f_{\beta}+L^{r^2}C^r_{\beta}}{\mathcal I_z\nu_x}\right)\dot{e}_\phi + \frac{L^fC^f_{\beta}}{\mathcal I_z}\theta^f \nonumber\\ 
&\,\,\,\,\,+ \left(-\frac{L^{f^2}C^f_{\beta}+L^{r^2}C^r_{\beta}}{\mathcal I_z\nu_x}\right)\dot{\phi}_{des}.
\end{align}
}
\normalsize
To facilitate modern control system design, these dynamics are reformulated in state-space form. Defining the state vector $\chi(t)=[e_d~ \dot{e}_d~ e_\phi~ \dot{e}_\phi]^T$ and considering the front wheel steering angle as the control input $u(t) = \theta^f$, one obtains:
\begin{equation}
\label{Eq__6}
\dot{\chi}(t) = \mathcal{H}\chi(t) + \mathcal{G}u(t),
\end{equation}
where the system matrices are explicitly defined as:

\small{
\begin{equation}
\label{Eq__7}
\mathcal{H}=\begin{bmatrix}
0 & 1 & 0 & 0 \\
0 & -\frac{C^f_{\beta}+C^r_{\beta}}{m\nu_x} & \frac{C^f_{\beta}+C^r_{\beta}}{m} & -\frac{L^fC^f_{\beta}-L^rC^r_{\beta}}{m\nu_x} \\
0 & 0 & 0 & 1 \\
0 & -\frac{L^fC^f_{\beta}-L^rC^r_{\beta}}{\mathcal I_z\nu_x} & \frac{L^fC^f_{\beta}-L^rC^r_{\beta}}{\mathcal I_z} & -\frac{L^{f^2}C^f_{\beta}+L^{r^2}C^r_{\beta}}{\mathcal I_z\nu_x}
\end{bmatrix}
\end{equation}}

\normalsize
\begin{equation}
\label{Eq__8}
\mathcal{G}=\begin{bmatrix}0 & \frac{C^f_{\beta}}{m} & 0 & \frac{L^fC^f_{\beta}}{\mathcal I_z}\end{bmatrix}^T.
\end{equation}

\subsection{State Measurement and Data Acquisition}
\label{sec:2.2}
The data-driven secure control framework integrates body sensors, steering actuators, and the ECU through a protected in-vehicle network architecture with safeguards against FDI attacks targeting the steering actuator subsystem. The state variables $e_d(t)$, $\dot{e}_d(t)$, $e_\phi(t)$, $\dot{e}_\phi(t)$, and $\theta^f(t)$ are measured and sampled at consistent intervals to ensure fidelity of the data-driven model. To validate the approach, experimental data from HORIBA MIRA's Autonomous Vehicle Development Centre in Nuneaton, UK, was utilized \cite{ccaar2025}. The dataset was acquired using an instrumented Range Rover research platform equipped with RTK GPS, inertial navigation system, LiDAR sensors, and high-precision steering angle encoders. Data collection involved multiple test runs at speeds between 15-30 m/s, capturing both standard lane-keeping maneuvers and evasive scenarios. The data acquisition system recorded samples at 100 Hz sampling frequency (sampling interval $\ell = 0.01$ s), from which a subset of 5000 samples was utilized. The measurement data comprising $m$ sequential samples is organized into structured matrices:
\begin{equation}
\label{Eq__11}
X =
\begin{bmatrix}
\chi(1) & \chi(2) & \cdots & \chi(m-1)
\end{bmatrix}_{n \times (m-1)},
\end{equation}
where $n$ is the dimension of the state vector $\chi(k)$. 

\begin{rem}
According to Willems' Fundamental Lemma \cite{willems2005note}, for the data-driven representation in \eqref{Eq__11} to be valid, the collected input sequence must be persistently exciting of sufficiently high order. Specifically, the dataset length \( m \) must satisfy $m \geq (n_s + 1)(\ell_s + 1) - 1$, where \( n_s \) is the system order and \( \ell_s \) is the window length or order of excitation. This ensures that the corresponding Hankel matrix has full row rank, enabling accurate reconstruction of all admissible system trajectories \cite{de2019formulas}.
\end{rem}

The corresponding time-shifted state matrix is given by:

\begin{equation}
\label{Eq__12}
X' =
\begin{bmatrix}
\chi(2) & \chi(3) & \cdots & \chi(m)
\end{bmatrix}_{n \times (m-1)}.
\end{equation}

To complete the experimental dataset, the control input sequence matrix can be constructed as:

\begin{equation}
\label{Eq__13}
\Xi =
\begin{bmatrix}
u(1) & u(2) & \cdots & u(m-1)
\end{bmatrix}_{p \times (m-1)},
\end{equation}
where $p$ is the dimension of the control input vector $u(k)$.

The temporal evolution of the lateral dynamics can be represented in discrete-time form:
\begin{equation}
\label{Eq__14}
\chi(k+1) = \mathcal{A}\chi(k) + \mathcal{B}u(k),
\end{equation}
where $\mathcal{A} \in \mathbb{R}^{n \times n}$ represents the state transition matrix and $\mathcal{B} \in \mathbb{R}^{n \times 1}$ denotes the input distribution matrix. 

It is essential to note that the system matrices $\mathcal{A}$ and $\mathcal{B}$, as described in the lateral dynamics model, are initially unknown and need to be identified through system identification techniques. Based on the experimental dataset provided in equations \eqref{Eq__12}-\eqref{Eq__14}, determining these matrices accurately is critical for representing the system's dynamics effectively and ensuring reliable control design.


\begin{rem}
\label{rem:3}
While the mechanism-based model provides a theoretical foundation, practical implementation requires discretization with sampling period \( \ell \) to obtain \eqref{Eq__6}. Note that event-triggered control implies variable transmission instants, not uniform sampling; \( \ell \) serves only as a reference for data-driven approximation. Since inherent nonlinearities, parametric uncertainties, and time-varying characteristics challenge precise mechanism-based modeling, a data-driven approach is adopted to approximate \eqref{Eq__6}. By complementing the reduced-order linear model, the data-driven approach ensures that the framework remains robust and practical, even in the presence of nonlinearities, thereby enhancing its applicability to real-world scenarios.
\end{rem}

\subsection{Efficient Event-Triggered Scheme}
\label{sec:2.3}

Efficient communication and robust control are key objectives in networked control systems for AVs. To achieve these goals, an event-triggered transmission scheme is introduced, which optimizes network resource usage by activating communication only when necessary. Accordingly, the sampling set is defined as $\mathcal{S}_1 = \{0, \ell, 2\ell, \ldots, k\ell, \ldots \}, \quad k \in \mathbb{N}$ where $\ell$ represents the fixed sampling period. The event-triggered transmission set, which contains the instants when data is transmitted, is a subset of \(\mathcal{S}_1\) defined as $\mathcal{S}_2 = \{k_0, k_1, k_2, \ldots, k_s, \ldots \}, \, s \in \mathbb{N}$, where all $k$s are taken from \(\mathcal{S}_1\). The event-triggered transmission mechanism is designed as:
\begin{align}
\label{Eq__15}
k_{s+1} &= \min\big\{k > k_s \, \big| \, \big[\chi(k) - \chi(k_s)\big]^T \Upsilon \big[\chi(k) - \chi(k_s)\big] \\ \nonumber
&\geq \mu \chi^T(k_s) \Upsilon \chi(k_s) \big\},
\end{align}
where \(k_s\) is the latest event-triggered instant, and \(k_{s+1}\) is the next triggered instant, and \(\chi(k_s)\) is the system state at time \(k_s\). The matrix \(\Upsilon\) is positive-definite time-invariant weighting matrix selected based on system dynamics, and \(\mu \in [0, 1]\) is the event-triggered parameter that determines the sensitivity of the triggering mechanism. During the intervals between consecutive transmissions, and for $\forall k \in [k_s, k_{s+1} - 1]$, the state error satisfies:
\begin{equation}
\label{Eq__16}
\big[\chi(k) - \chi(k_s)\big]^T \Upsilon \big[\chi(k) - \chi(k_s)\big] \leq \mu \chi^T(k_s) \Upsilon \chi(k_s),
\end{equation}
and the error dynamics between the sampled and transmitted states are defined as:
\begin{equation}
e(k) = \chi(k) - \chi(k_s),
\label{Eq__17}
\end{equation}
where \(e(k)\) represents the error vector. 

Substituting \eqref{Eq__17} into \eqref{Eq__16}, the triggering condition becomes:
\begin{equation}
\label{Eq__18}
e(k)^T \Upsilon e(k) \leq \mu \chi^T(k_s) \Upsilon \chi(k_s).
\end{equation}

\begin{rem}
\label{rem:3_2}
Substituting \eqref{Eq__17} into \eqref{Eq__16} yields the bounded error condition $\|e(k)\|^2 \leq \mu \|\chi(k_s)\|^2$. Using spectral properties of $\Upsilon$, the error magnitude is bounded by $\|e(k)\| \leq \sqrt{\mu} \|\chi(k_s)\|$, providing a quantitative measure for event sensitivity. The parameter \(\mu\) must be carefully selected to balance communication efficiency and control performance. Smaller \(\mu\) ensures frequent transmissions, enhancing accuracy but increasing overhead, while larger \(\mu\) reduces communication frequency, improving efficiency but potentially degrading performance and delaying responses to system changes. This trade-off highlights the importance of tuning \(\mu\) based on application-specific requirements.
\end{rem}

\subsection{Extended State Observer for Actuator Attack Detection and Compensation}
\label{sec:2.4}

In practical scenarios, networked control systems may be subject to actuator attacks, which can compromise system integrity. Considering the data-driven model \eqref{Eq__14}, the attacks are modeled as follows:
\begin{align}
\label{Eq__19}
\chi(k+1) &= \mathcal{A}\chi(k) + \mathcal{B}\tilde{u}(k), \nonumber \\
\tilde{u}(k) &= u(k) + \alpha_{att}(k),
\end{align}
where \(\tilde{u}(k)\) is the corrupted input due to the FDI attack signal \(\alpha_{att}(k)\). Substituting \(\tilde{u}(k)\) into \eqref{Eq__19} yields:
\begin{equation}
\label{Eq__20}
\chi(k+1) = \mathcal{A}\chi(k) + \mathcal{B}u(k) + \mathcal{B}\alpha_{att}(k).
\end{equation}

\begin{assump}
\label{ass:1}
It is assumed that $\alpha_{att}(k)$ is a slow-varying disturbance, and its magnitude remains bounded such that $\|\alpha_{att}(k)\| \leq \mathcal{Q}_{att}$, where $\mathcal{Q}_{att}$ represents the upper bound of the attack magnitude within the range of $0.05\leq \mathcal{Q}_{att}\leq 0.2$.
\end{assump}


\begin{rem}
\label{rem:555}
While Assumption \ref{ass:1} requires $\|\alpha_{att}(k)\| \leq Q_{att}$ as a necessary condition for theoretical guarantees, it is important to consider implications when this bound is violated. If $\|\alpha_{att}(k)\| > Q_{att}$, attack compensation effectiveness degrades as the actual attack magnitude exceeds the design parameter. The sliding surface experiences increased deviations from zero, and tracking performance deteriorates in proportion to the extent to which the actual attack exceeds $Q_{att}$. The ESO-based estimation may still track the attack pattern but with reduced accuracy, potentially leading to partial compensation. However, as attack magnitude increases significantly beyond $Q_{att}$, the theoretical stability guarantees of Theorem \ref{thm:2} no longer hold, and the system could potentially become unstable. This limitation highlights an important design trade-off in $Q_{att}$: setting it too low may leave the system vulnerable to stronger attacks, while setting it too high may lead to conservative control actions and potential chattering.
\end{rem}

\begin{assump}
\label{ass:2}
The disturbance affecting the system dynamics, including sensor noise and unmodeled dynamics, is assumed to be bounded within \(\|w(k)\| \leq \Gamma\), ensuring that the ESO's performance is not significantly degraded by external noise.
\end{assump}

To estimate the attack signal \(\alpha_{att}(k)\), an ESO is developed. Let us define the augmented state \(\zeta(k)\) as:
\begin{equation}
\zeta(k) = \begin{bmatrix} \chi(k) \\ \alpha_{att}(k) \end{bmatrix},
\label{Eq__21}
\end{equation}
with the corresponding augmented dynamics $\zeta(k+1) = \begin{bmatrix} \mathcal{A} & \mathcal{B} \\ 0 & I \end{bmatrix} \zeta(k) + \begin{bmatrix} \mathcal{B} \\ 0 \end{bmatrix} u(k)$. The ESO can then be designed as:
\begin{equation}
\hat{\zeta}(k+1) = \begin{bmatrix} \mathcal{A} & \mathcal{B} \\ 0 & I \end{bmatrix} \hat{\zeta}(k) + \begin{bmatrix} \mathcal{B} \\ 0 \end{bmatrix} u(k) + L_\zeta\big(y(k) - \hat{y}(k)\big),
\label{Eq__23}
\end{equation}
where \(L_\zeta\) is the observer gain matrix for the augmented system.

\begin{rem}
\label{rem:5}
The design of \(L_\zeta\) must ensure that the augmented error dynamics, defined as \(e_\zeta(k) = \zeta(k) - \hat{\zeta}(k)\), remain stable. The error dynamics are governed by:
\[
e_\zeta(k+1) = \Big(\begin{bmatrix} \mathcal{A} & \mathcal{B} \\ 0 & I \end{bmatrix} - L_\zeta C_\zeta\Big)e_\zeta(k),
\]
where \(C_\zeta = \begin{bmatrix} C & 0 \end{bmatrix}\) with $C \in \mathbb{R}^{m \times n}$. Stability is guaranteed if \(\Big\|\begin{bmatrix} \mathcal{A} & \mathcal{B} \\ 0 & I \end{bmatrix} - L_\zeta C_\zeta\Big\| < 1\).
\end{rem}

The attack estimation is then extracted as $\hat{\alpha}_{att}(k) = \begin{bmatrix} 0 & I \end{bmatrix} \hat{\zeta}(k)$. Given the estimation \(\hat{\alpha}_{att}(k)\), the control input can be redesigned as:
\begin{equation}
u(k) = K\chi(k) - \hat{\alpha}_{att}(k),
\label{Eq__24}
\end{equation}
where \(K\) denotes the state-feedback gain matrix. 

Substituting \eqref{Eq__24} into \eqref{Eq__20} yields the compensated system dynamics as:
\begin{equation}
\chi(k+1) = (\mathcal{A} + \mathcal{B}K)\chi(k) + \mathcal{B}\big(\alpha_{att}(k) - \hat{\alpha}_{att}(k)\big).
\label{Eq__25}
\end{equation}

In addition to ensuring stability through the design of \(L_\zeta\) as described in Remark \ref{rem:5}, the augmented error dynamics should be carefully analyzed to ensure rapid convergence of the estimated states. The error dynamics, defined as \(e_\zeta(k) = \zeta(k) - \hat{\zeta}(k)\), evolve as 
\begin{equation*}
e_\zeta(k+1) = \Big(\begin{bmatrix} \mathcal{A} & \mathcal{B} \\ 0 & I \end{bmatrix} - L_\zeta C_\zeta\Big)e_\zeta(k),
\end{equation*}
where \(C_\zeta = \begin{bmatrix} C & 0 \end{bmatrix}\). For practical implementations, the choice of \(L_\zeta\) must guarantee stability (as per \(\|\mathcal{A} - L_\zeta C\| < 1\)) and achieve a balance between estimation speed and noise rejection. A larger \(L_\zeta\) may lead to faster convergence but amplify noise effects, whereas a smaller \(L_\zeta\) may introduce delays in attack detection.

\begin{rem}
\label{rem:7}
The ESO design assumes that the FDI attack signal \(\alpha_{att}(k)\) evolves as a persistent or slowly varying parameter, reflecting practical FDI scenarios designed to evade detection. While the framework primarily addresses such attack dynamics, the proposed FSMC provides natural rejection capabilities against bounded external disturbances. General disturbances (sensor noise, environmental uncertainties) are modeled as bounded signals (Assumption \ref{ass:2}), while the attack threshold parameter $\mathcal{Q}_{att}$ (Assumption \ref{ass:1}) distinguishes nominal disturbances from malicious interventions. The ESO enables simultaneous estimation of both system state and attack signal, enhancing system resilience when direct attack measurement is infeasible.
\end{rem}

\section{Data-Driven Secure Event-Triggered Control Framework}
\label{sec:3}

This section focuses on developing secure event-triggered control for the networked lateral regulation of AVs using the data-driven model (whose framework is depicted in Fig. \ref{fig:0}). The proposed framework addresses transmission delays, event-triggered communication, and actuator attacks, ensuring both stability and robustness of the control system. Transmission delays are an inherent characteristic of networked control systems, especially when communication is event-triggered. Let \(\delta_k\) represent the delay at time \(k\), where \(\delta_k \in (0, \delta]\) and \(\delta \in \mathbb{N}\). Two primary cases are considered for analyzing and compensating for delays.

\begin{figure}[!t]
\centering
\includegraphics[width=3.5 in]{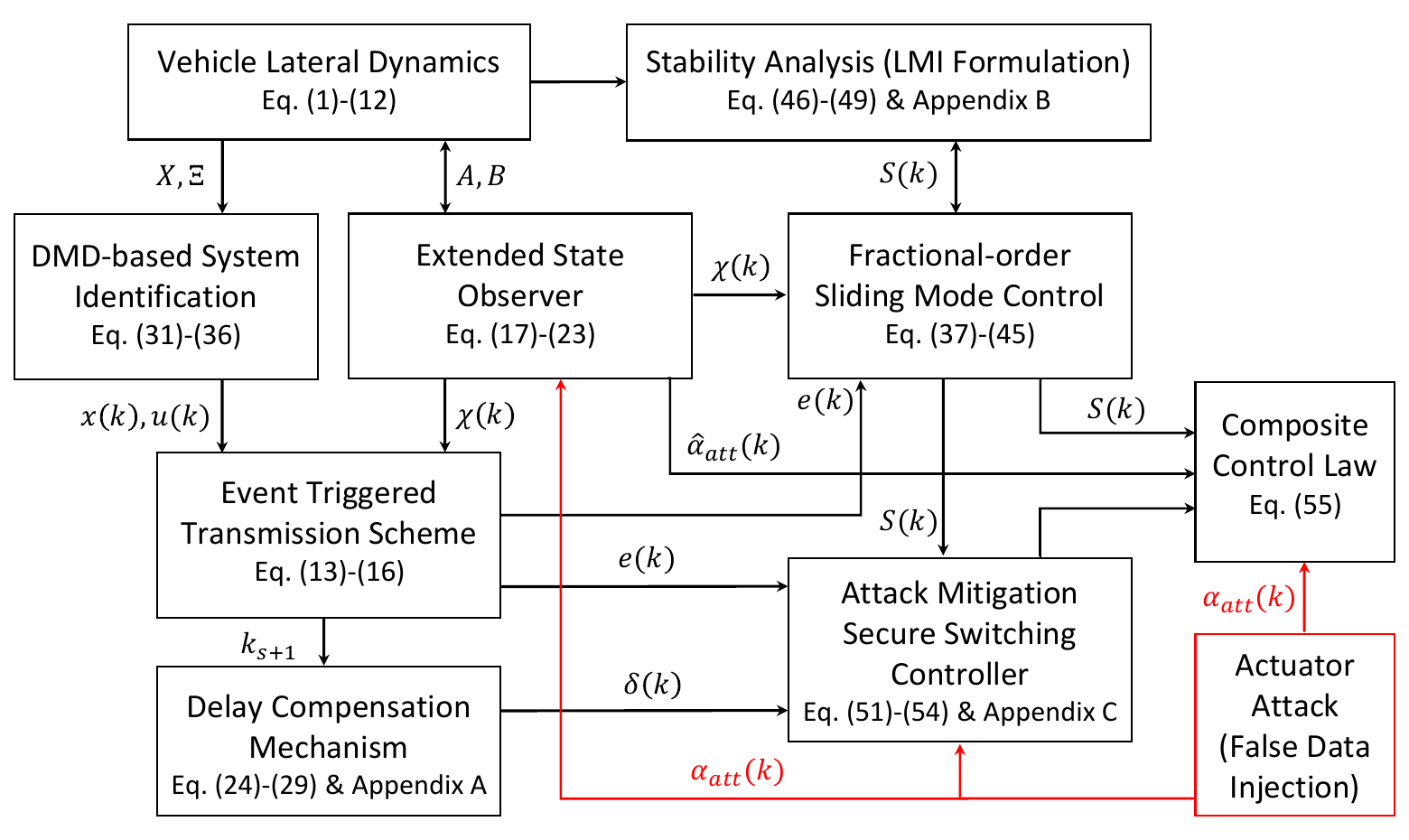}
\caption{Cyber-resilient data-driven ETS-FSMC control framework for AVs under FDI attacks.}
\label{fig:0}
\end{figure}

\subsubsection{Case A - Transmission Delay Within Triggered Instants}
\label{sec:3.1.1}
If \(k_s + \delta + 1 \geq k_{s+1} + \delta_{k_{s+1}} - 1\), an artificial delay \(\delta(k)\) is introduced to align with the event-triggered transmission scheme as
\begin{equation}
\label{Eq__26}
\delta(k) = k - k_s, \quad k \in \big[k_s + \delta_{k_s}, k_{s+1} + \delta_{k_{s+1}} - 1\big],
\end{equation}
that satisfies $\delta_{k_s} \leq \delta(k) \leq k_{s+1} - k_s + \delta_{k_{s+1}} - 1 \leq 1 + \delta$. The bounded nature of \(\delta(k)\) ensures that the control system operates within predictable limits despite variations in transmission delays.

\subsubsection{Case B - Transmission Delay Beyond Triggered Instants}
\label{sec:3.1.2}
If \(k_s + \delta + 1 < k_{s+1} + \delta_{k_{s+1}} - 1\), and having \([k_s + \delta_{k_s}, k_s + \delta]\) and \([k_s + \delta + l, k_s + \delta + l + 1]\), \(l \in \mathbb{Z}^+\), where \(l \geq 1\), $\exists d \; \text{s.t.}\, k_s + d + \delta < k_{s+1} + \delta_{k_{s+1}} - 1 \leq k_s + d + \delta + 1$. Given that \(\chi(k_s)\) and \(\chi(k_s+l)\) with \(l = 0, 1, 2, \dots, d-1\) satisfy the event-triggered condition \eqref{Eq__16}, the time intervals can be divided into three distinct segments to ensure accurate modeling of delays in scenarios where the delay exceeds a single transmission interval as follows:
\begin{align}
\label{Eq__27}
\begin{cases}
\Lambda_1 = [k_s + \delta_{k_s}, k_s + \delta + 1], \\
\Lambda_2 = [k_s + \delta + l, k_s + \delta + l + 1], \, l \in \mathbb{Z}^+, \\
\Lambda_3 = [k_s + \delta + d, k_{s+1} + \delta_{k_{s+1}} - 1].
\end{cases}
\end{align}
where \([k_s + \delta_{k_s}, k_{s+1} + \delta_{k_{s+1}} - 1]\) can be represented as $[k_s + \delta_{k_s}, k_{s+1} + \delta_{k_{s+1}} - 1] = \bigcup_{i=1}^3 \Lambda_i$. 

The artificial delay \(\delta(k)\) for these intervals can be defined as:
\begin{equation}
\label{Eq__28}
\delta(k) =
\begin{cases}
k - k_s, & k \in \Lambda_1, \\
k - k_s - l, & k \in \Lambda_2, \\
k - k_s - d, & k \in \Lambda_3,
\end{cases}
\end{equation}
which leads to
\begin{equation}
\label{Eq__29}
\begin{cases}
\delta_{k_s} \leq \delta(k) \leq 1 + \delta = \bar{\delta}, & k \in \Lambda_1, \\
\delta_{k_s} < \delta(k) \leq \bar{\delta}, & k \in \Lambda_2, \\
\delta_{k_s} < \delta(k) \leq \bar{\delta}, & k \in \Lambda_3.
\end{cases}
\end{equation}

Accordingly, one has \(0 \leq \delta_{k_s} \leq \delta(k) \leq \bar{\delta}\). To manage state deviations due to delays, the error vector \(e(k)\) is defined as  \(e(k) = 0\) for \(k \in [k_s + \delta_{k_s}, k_s+1 + \delta_{k_s+1} - 1]\) for Case A, and 
\begin{equation}
\label{Eq__30}
e(k) =
\begin{cases}
0, & k \in \Lambda_1, \\
\chi(k_s) - \chi(k_s + l), & k \in \Lambda_2, \\
\chi(k_s) - \chi(k_s + d), & k \in \Lambda_3,
\end{cases}
\end{equation}
for Case B. 
\begin{rem}
\label{rem:8_2}
The artificial delay $\delta(k)$ introduces a controlled variation in state dynamics, ensuring that network delays do not disrupt control objectives. However, this approach may amplify errors if communication delays exceed $\bar{\delta}$.
 \end{rem}

\begin{thm}
\label{thm:1}
Given \(e(k)\) in \eqref{Eq__17} and the event-triggered scheme \eqref{Eq__16}, the error dynamics satisfy:
\begin{align}
\label{Eq__31}
e^T(k)\Upsilon e(k) & \leq \mu \chi(k - \delta(k))^T \Upsilon \chi(k - \delta(k)), \nonumber \\
& \text{for} \quad k \in [k_s + \delta_{k_s}, k_{s+1} + \delta_{k_{s+1}} - 1],
\end{align}
where \(\Upsilon\) is a positive-definite weighting matrix, and \(\mu \in [0, 1]\) controls the triggering sensitivity. Accordingly, the stability of the event-triggered control system is ensured if the error vector \(e(k)\) satisfies the event-triggered condition \eqref{Eq__16}.
\end{thm}

The proof is provided in Appendix~\ref{appendix:1}.

\begin{rem}
\label{rem:9}
The Lyapunov-based analysis confirms the exponential decay of the error vector \(e(k)\), ensuring robust stability under the event-triggered control scheme. The exponential rate \(\mu^k\) and the matrix \(\Upsilon\) are pivotal design parameters that significantly influence system performance. The selection of the weighting matrix \(\Upsilon\) directly determines the convergence bounds of \(e(k)\), offering flexibility in shaping error dynamics. Similarly, the parameter \(\mu\) governs the triggering sensitivity, where a smaller \(\mu\) accelerates convergence but results in a higher frequency of triggering events, thereby increasing communication overhead. Consequently, careful tuning of both \(\Upsilon\) and \(\mu\) facilitates an optimal trade-off between control performance and communication efficiency, making the proposed framework highly effective in optimizing networked control systems.
\end{rem}

The dynamics of the system under consideration, which includes event-triggered transmission, actuator attacks, and artificial delays, can be expressed as:
\begin{equation}
\label{Eq__35}
\chi(k+1) = \mathcal{A}\chi(k) + \mathcal{B}\mathcal{K}\chi(k - \delta(k)) + \mathcal{B}\mathcal{K}e(k) + \mathcal{B}\alpha_{att}(k),
\end{equation}
 and the event-triggered transmission is activated according to \eqref{Eq__15} subjected to \eqref{Eq__31}. It is worth noting that \eqref{Eq__15}, \eqref{Eq__16}, and \eqref{Eq__31}  describe similar event-triggered control scheme, providing complementary insights into its operation. In this context, \eqref{Eq__15} defines the event-generation instant \(k_{s+1}\) as the minimum time index satisfying the condition \([\chi(k) - \chi(k_s)]^T \Upsilon [\chi(k) - \chi(k_s)] \geq \mu \chi^T(k_s) \Upsilon \chi(k_s)\). On the other hand, \eqref{Eq__16} specifies the event-triggered condition, ensuring that \([\chi(k) - \chi(k_s)]^T \Upsilon [\chi(k) - \chi(k_s)] \leq \mu \chi^T(k_s) \Upsilon \chi(k_s)\) holds during each interval between triggered instants. Furthermore, \eqref{Eq__31} characterizes the time delay at each triggered instant by incorporating \eqref{Eq__28} and \eqref{Eq__30}, providing a unified description of the delay dynamics.

\begin{rem}
\label{rem:11}
The event-triggered scheme effectively reduces unnecessary data transmission, thereby optimizing the use of communication resources in networked control systems. Furthermore, the actuator attack signal \(\alpha_{att}(k)\) is compensated for using the control law, ensuring robust performance even under malicious interventions. By segmenting delays into distinct cases, the framework ensures precise modeling and compensation of varying delay conditions, maintaining system stability and performance.
\end{rem}


\section{Main Results and System Implementation}
\label{sec:4}
This section presents the secure control design for the developed event-triggered lateral control of AVs based on the data-driven model.

\subsection{Data-Driven Modeling via Dynamic Model Decomposition}
\label{sec:4.1}

To determine the system matrices \(\mathcal{A}\) and \(\mathcal{B}\) from the collected datasets \(X'\), \(X\), and \(\Xi\), the DMD technique is employed, which ensures accurate identification of the unknown dynamics, enabling the design of precise control strategies. Using the relationships defined in \eqref{Eq__11}--\eqref{Eq__14}, the dynamics can be expressed as:
\begin{equation}
X' \approx \Gamma \Theta,
\label{Eq__36}
\end{equation}
where $\Gamma = \begin{bmatrix} \mathcal{A} & \mathcal{B} \end{bmatrix}$, and $\Theta = \begin{bmatrix} X & \Xi \end{bmatrix}^T$, with \(\Gamma \in \mathbb{R}^{n \times (n+1)}\) and \(\Theta \in \mathbb{R}^{(n+1) \times (m-1)}\). 

Considering \eqref{Eq__36}, the system operators can be approximated as
\begin{align}
\label{Eq__37}
\Gamma &= X' \Theta^\dagger, \\
\begin{bmatrix}
\mathcal{A} & \mathcal{B}
\end{bmatrix}
&= X' \begin{bmatrix}
X \\ \Xi
\end{bmatrix}^\dagger,
\end{align}
where \(\dagger\) represents the Moore–Penrose pseudoinverse of the matrix.

Accordingly, to improve the robustness and computational efficiency of the decomposition process, the singular value decomposition (SVD) is applied to \(\Theta\), expressed as:
\begin{align}
\label{Eq__38}
\Theta &= U \Sigma V^H \\ \nonumber
&= \begin{bmatrix}
\tilde{U} & U_\text{rem}
\end{bmatrix}
\begin{bmatrix}
\tilde{\Sigma} & 0 \\ 0 & \Sigma_\text{rem}
\end{bmatrix}
\begin{bmatrix}
\tilde{V}^H \\ V_\text{rem}^H
\end{bmatrix} \approx \tilde{U} \tilde{\Sigma} \tilde{V}^H,
\end{align}
where \(U \in \mathbb{R}^{(n+1) \times (n+1)}\) and \(V \in \mathbb{R}^{(m-1) \times (m-1)}\) are orthogonal matrices, \(\Sigma \in \mathbb{R}^{(n+1) \times (m-1)}\) is a diagonal matrix containing singular values, \(r\) represents the truncation order, and \(\tilde{U}\), \(\tilde{\Sigma}\), and \(\tilde{V}^H\) are truncated components corresponding to the dominant singular values.

The truncated approximation allows the refinement of \(\Gamma\) as:
\begin{equation}
\Gamma \approx \tilde{\Gamma} = X' \tilde{V} \tilde{\Sigma}^{-1} \tilde{U}^H.
\label{Eq__39}
\end{equation}

\(\mathcal{A}\) and \(\mathcal{B}\) can then be extracted from the decomposed operator \(\tilde{\Gamma}\) as:
\begin{align}
\label{Eq__40}
\begin{bmatrix}
\mathcal{A}, \mathcal{B}
\end{bmatrix}
&\approx
\begin{bmatrix}
\tilde{\mathcal{A}}, \tilde{\mathcal{B}}
\end{bmatrix} \\ \nonumber
&=
\begin{bmatrix}
X' \tilde{V} \tilde{\Sigma}^{-1} \tilde{U}_1^H, \, X' \tilde{V} \tilde{\Sigma}^{-1} \tilde{U}_2^H
\end{bmatrix},
\end{align}
where \(\tilde{U}_1^H \in \mathbb{R}^{r \times n}\) and \(\tilde{U}_2^H \in \mathbb{R}^{r \times 1}\) represent the partitioned left singular variables.

\begin{rem}
\label{rem:12}
The choice of truncation order \(r\) in SVD significantly influences the accuracy and computational load of the identified model \cite{gavish5870optimal}. While retaining more singular values improves precision, it also increases computational demand. Accordingly, the truncation order should be chosen to balance accuracy with efficiency based on the application’s requirements.
\end{rem}

\begin{rem}
\label{rem:13}
The selection of the order of an unknown system in data-driven modeling presents a significant challenge. However, by leveraging insights from the mechanism-based model, the order of the studied system can be determined systematically. In this context, the state-space model presented in this work provides guidance, where the primary state variables, such as \(e_d(t)\), \(\dot{e}_d(t)\), \(e_\phi(t)\), and \(\dot{e}_\phi(t)\), serve as a basis for selecting the appropriate order and the corresponding sampling states. Furthermore, the identification \(\mathcal{A}\) and \(\mathcal{B}\) critically depends on solving \eqref{Eq__37}, which relies on the collected datasets \(X'\), \(X\), and \(\Xi\). Typically, solving \(X' \approx \Gamma \Theta\) constitutes an underdetermined problem. To address this, a least-squares solution for \(\Gamma\) can be obtained by minimizing the Frobenius norm of the residual \(\|X' - \Gamma \Theta\|_F\), ensuring a consistent and reliable approximation.
\end{rem}

Algorithm \ref{alg:1} outlines the systematic steps for applying DMD to identify the lateral dynamics of AVs. In Step 1 of Algorithm \ref{alg:1}, selecting a smaller sampling period \(\ell\) improves the fidelity of the identified model but increases computational costs. It is essential to balance the sampling rate with the available computational resources. Additionally, maintaining consistency in the sampling period for both modeling and control implementation ensures seamless integration.

\begin{algorithm}[t!]
\caption{DMD for Lateral Dynamics}
\footnotesize
\label{alg:1}
\begin{algorithmic}[1]
\Require Sampling period \(\ell\), input dataset \(\Xi\), collected state datasets \(X, X'\).
\Ensure Estimated system matrices \(\mathcal{A}\) and \(\mathcal{B}\) for control implementation.
\State \textbf{Step 1: Data Acquisition and Preprocessing}
\Statex  Obtain the state dataset \(X\) and its time-shifted counterpart \(X'\) from measured vehicle states.
\State \textbf{Step 2: Singular Value Decomposition (SVD)}
\Statex  Construct the data matrix \(\Theta = \begin{bmatrix} X & \Xi \end{bmatrix}^T\).
\Statex  Compute the SVD of \(\Theta\) as \(\Theta = U\Sigma V^H\) and retain dominant singular values up to truncation order \(r\).
\State \textbf{Step 3: System Identification}
\Statex  Approximate the operator \(\Gamma = \begin{bmatrix} \mathcal{A} & \mathcal{B} \end{bmatrix}\) using \(\Gamma \approx X' \tilde{V} \tilde{\Sigma}^{-1} \tilde{U}^H\) as per \eqref{Eq__39}.
\Statex Extract \(\mathcal{A}\) and \(\mathcal{B}\) from the decomposed operator following \eqref{Eq__40}.

\State \textbf{Step 4: Output and Storage}
\Statex Store the estimated system matrices \(\mathcal{A}, \mathcal{B}\) for subsequent control design.
\end{algorithmic}
\end{algorithm}

\subsection{Secure Control Design and Stability Analysis}
\label{sec:4.2}
The core concept of the proposed secure control approach consists of two interconnected components. The first involves developing an inherent equivalent controller, \(u_e(k)\), which ensures stability for the nominal system in the absence of actuator attacks. The second focuses on designing a switching controller, \(u_s(k)\), specifically intended to counteract and mitigate the effects of actuator attacks. Given that, the secure SMC structure can be defined as $u(k) = u_e(k) + u_s(k)$. Hence, the FO sliding surface is proposed as

\begin{equation}
\label{Eq__41}
S(k) = \mathcal{F}\chi(k) +   \epsilon(k)+ \lambda_2 \mathcal{D}^\gamma \epsilon(k),
\end{equation}
where $S(0) = 0$, $\mathcal{F} = \mathcal{B}^TP$, \(\epsilon(k)\) represents the error dynamics, $\epsilon(k+1) - \epsilon(k) = \mathcal{F} \chi(k) - \mathcal{F} (\mathcal{A} + \mathcal{B}\mathcal{K})\chi(k)$,  \(\mathcal{D}^\gamma\) is the fractional derivative operator with order \( \gamma \in (0, 1)\), and \(\lambda > 0\) is a tuning parameter. $\mathcal{D}^\gamma \epsilon(k)$ denotes the discrete-time Grünwald-Letnikov fractional derivative \cite{mousavi2023observer}.


It is worth noting that, introducing the matrix \(\mathcal{F}\) in the sliding mode surface \eqref{Eq__41} is fundamental to ensuring the desired system behavior under the proposed control framework, where \(P\) is obtained by solving the Riccati equation $P = \mathcal{A}^T P \mathcal{A} - \mathcal{A}^T P \mathcal{B} \big(R + \mathcal{B}^T P \mathcal{B}\big)^{-1} \mathcal{B}^T P \mathcal{A} + Q$, where \(P\), \(Q\), and \(R\) are symmetric positive definite matrices with appropriate dimensions. \( Q \) is suitably chosen to prioritize the minimization of lateral position error and yaw rate deviations, where higher weights are assigned to the lateral error state components to enforce precise trajectory tracking. On the other hand, \( R \) penalizes excessive control effort to prevent aggressive steering commands, ensuring actuator constraints are respected. It should be noted that the term \(\big(R + \mathcal{B}^T P \mathcal{B}\big)^{-1}\) represents the inverse of the cost-weighted control effect, balancing the trade-off between control effort and state regulation. To elaborate further, the role of \(\mathcal{F}\) in the sliding mode surface can be understood as a projection operator that determines how the state vector \(\chi(k)\) influences the surface dynamics. By substituting \(\mathcal{F} \) into \eqref{Eq__41}, the sliding mode surface becomes:
\begin{equation}
\label{Eq__42}
S(k) = \mathcal{B}^T P \chi(k) +  \epsilon(k) + \lambda_2 \mathcal{D}^\gamma \epsilon(k).
\end{equation}

In this context, the computation of \(P\) through the Riccati equation ensures that the system remains stable by minimizing a quadratic cost function associated with the state and control effort. The matrix \(P\) essentially captures the cost-to-go in an optimal control sense, providing a systematic way to design \(\mathcal{F}\) for achieving robust sliding mode dynamics.

\begin{thm}
\label{thm:2}
Consider the data-driven control model under actuator attacks as defined in \eqref{Eq__20} and governed by the event-triggered scheme in \eqref{Eq__15}. The system can be stabilized using a control law consisting of an equivalent state feedback controller $u_e(k) = -\kappa \operatorname{sgn}(S(k)) + \mathcal{K}\chi(k_s)$ to address nominal dynamics and a switching controller \(u_s(k_s) = -\alpha_{att}(k)\) to mitigate actuator attacks, ensuring overall system stability.
\end{thm}

\begin{proof}
As mentioned in Assumption \ref{ass:1}, the boundedness condition $\|\alpha_{att}(k)\| \leq \mathcal{Q}_{att}, \quad \forall k \geq 0,$ is considered, ensuring that the switching control law $u_s(k) = -\alpha_{att}(k)$ effectively cancels actuator faults. In this context, substituting \( u_s(k) \) into \( u(k) \) yields $u(k) = -\kappa \operatorname{sgn}(S(k)) + \mathcal{K} \chi(k_s) - \alpha_{att}(k)$. Since \( u_s(k) \) is designed to be equal and opposite to \( \alpha_{att}(k) \), the actuator fault is effectively neutralized, leading to the simplified control input $u(k) = -\kappa \operatorname{sgn}(S(k)) + \mathcal{K} \chi(k_s)$.

Then, considering the stability condition of the discrete-time SMC scheme, the sliding surface satisfies $S(k+1) - S(k) = 0$. Expanding \(S(k+1) - S(k)\) using the definition of the sliding surface with the fractional derivative term yields,
\begin{align}
\label{Eq__43}
S(k+1) - S(k) &= \mathcal{F}\chi(k+1) - \mathcal{F}\chi(k) + \epsilon(k+1) - \epsilon(k) \nonumber \\
&\quad+ \lambda_2 (\mathcal{D}^\gamma \epsilon(k+1) - \mathcal{D}^\gamma \epsilon(k)) \nonumber \\
&= \mathcal{F}\big[\mathcal{A}\chi(k) + \mathcal{B}u(k) + \mathcal{B}\alpha_{att}(k)\big] - \mathcal{F}\chi(k) \nonumber \\
&\quad + \mathcal{F}\chi(k) - \mathcal{F}(\mathcal{A} + \mathcal{B}\mathcal{K})\chi(k) \nonumber \\
&\quad+ \lambda_2 (\mathcal{D}^\gamma \epsilon(k+1) - \mathcal{D}^\gamma \epsilon(k)) \nonumber \\
&= \mathcal{F}\mathcal{B}u(k) + \mathcal{F}\mathcal{B}\alpha_{att}(k) - \mathcal{F}\mathcal{B}\mathcal{K}\chi(k) \nonumber \\
&\quad+ \lambda_2 (\mathcal{D}^\gamma \epsilon(k+1) - \mathcal{D}^\gamma \epsilon(k)).
\end{align}

To satisfy \(S(k+1) - S(k) = 0\), the control law \(u(k)\) must ensure:
\begin{align}
\label{Eq__44}
\mathcal{F}\mathcal{B}u(k) &+ \mathcal{F}\mathcal{B}\alpha_{att}(k) - \mathcal{F}\mathcal{B}\mathcal{K}\chi(k)  \nonumber \\ 
&+ \lambda_2 (\mathcal{D}^\gamma \epsilon(k+1) - \mathcal{D}^\gamma \epsilon(k)) = 0.
\end{align}

Substituting \( u(k) \), one obtains
\begin{align}
\label{Eq__44a}
&-\mathcal{F}\mathcal{B} \kappa \operatorname{sgn}(S(k)) + \mathcal{F}\mathcal{B} \mathcal{K} (\chi(k_s) - \chi(k)) \nonumber \\
&+ \mathcal{F}\mathcal{B} \alpha_{att}(k) + \lambda_2 (\mathcal{D}^\gamma \epsilon(k+1) - \mathcal{D}^\gamma \epsilon(k)) = 0.
\end{align}

Under the event-triggered scheme, \(u(k)\) can be rewritten as:
\begin{align}
\label{Eq__45}
u(k) &= -\kappa \operatorname{sgn}(S(k)) + \mathcal{K}\chi(k_s) + u_s(k) \nonumber \\
&= -\kappa \operatorname{sgn}(S(k)) + \mathcal{K}\chi(k - \delta(k)) + \mathcal{K}e(k) + u_s(k),
\end{align}
where \(\delta(k)\) represents the communication delay and \(e(k) = \chi(k) - \chi(k_s)\) is the error between the transmitted and current states.

Substituting \eqref{Eq__45} into \eqref{Eq__19} and incorporating the fractional derivative term yields
\begin{align}
\chi(k+1) &= \mathcal{A}\chi(k) + \mathcal{B}\mathcal{K}\chi(k - \delta(k)) + \mathcal{B}\mathcal{K}e(k) \nonumber \\
&\quad - \mathcal{B}\kappa \operatorname{sgn}(S(k)) + \lambda_2 \mathcal{D}^\gamma \epsilon(k).
\label{Eq__46}
\end{align}

To stabilize the system's lateral dynamics in the presence of actuator attacks, the switching controller \( u_s(k_s) = -\alpha_{att}(k) \) is employed to cancel the impact of the attack signal directly. Combining this with the equivalent controller \( u_e(k) = -\kappa \operatorname{sgn}(S(k)) + \mathcal{K}\chi(k_s) \), the overall system dynamics are stabilized, ensuring the convergence of the sliding surface \( S(k) \) to zero.
\end{proof}

\begin{rem}
\label{rem:155}
It is worth noting that, while Theorem \ref{thm:2} establishes that $u_s(k) = -\alpha_{att}(k)$ effectively cancels the attack impact, in practical implementations, the estimated attack signal $\hat{\alpha}_{att}(k)$ will inevitably contain estimation errors. These errors result in residual attack effects that can be expressed as $\tilde{\alpha}_{att}(k) = \alpha_{att}(k) - \hat{\alpha}_{att}(k)$. This residual attack term acts as a bounded disturbance on the system, where $\|\tilde{\alpha}_{att}(k)\| \leq \epsilon_{est}$ with $\epsilon_{est}$ representing the upper bound of estimation error. The proposed FSMC inherently provides robustness against such bounded disturbances through its switching term. By selecting the controller parameter $\kappa$ to satisfy $\kappa > \|FB\tilde{\alpha}_{att}(k)\|$, the sliding surface convergence can still be guaranteed despite estimation errors.
\end{rem}

\begin{rem}
\label{rem:16}
Theorem \ref{thm:2} establishes the groundwork for developing a secure control strategy for the lateral dynamics of AVs. However, it does not provide explicit details on the design of the feedback gain \(\mathcal{K}\) in the equivalent controller \(u_e(k)\) or the switching controller \(u_s(k)\). Notably, \(u_e(k)\) and \(u_s(k)\) can be formulated independently based on the principles outlined in Theorem \ref{thm:2}. Unlike conventional approaches, the updated control law incorporates a discontinuous reaching law term \( -\kappa \operatorname{sgn}(S(k)) \), which enhances robustness against uncertainties and disturbances. This modification ensures that the sliding surface \( S(k) \) is driven to zero in finite time, reinforcing stability against external perturbations and modeling inaccuracies. Moreover, the nominal event-triggered control system described in \eqref{Eq__46} inherently ensures the stability of the lateral dynamics \eqref{Eq__20}, provided that the switching controller effectively compensates for actuator attacks by setting \(u_s(k_s) = -\alpha_{att}(k)\). The subsequent Theorem \ref{thm:3} elaborates on the design criteria for the event-triggered parameter \(\mu\), ensuring the stability of the nominal event-triggered control system under the given conditions.
\end{rem}

\subsubsection{Stability Analysis}
\label{sec:4.2.1}
The stability of the proposed FSMC system can be analyzed using a Lyapunov candidate function:
\begin{equation}
\label{Eq__46a}
\mathcal{V}(k) = \frac{1}{2} S^2(k),
\end{equation}
where \(\mathcal{V}(k)\) is a positive definite function. Considering the approximation \( S(k+1) + S(k) \approx 2S(k) \), which is valid in the neighborhood of the discrete sliding surface~\cite{samantaray2020discrete}, the discrete-time difference yields
\begin{align}
\label{Eq__46b}
\Delta \mathcal{V}(k) &= \mathcal{V}(k+1) - \mathcal{V}(k) \nonumber \\
&= S(k)\big(S(k+1) - S(k)\big).
\end{align}

Substituting \eqref{Eq__41} into \eqref{Eq__46b} yields:
\begin{align}
\Delta \mathcal{V}(k) &= S(k)\Big[\mathcal{F}\big[\chi(k+1) - \chi(k)\big] \nonumber \\
& \quad + \lambda_2 \big[\mathcal{D}^\gamma \epsilon(k+1) - \mathcal{D}^\gamma \epsilon(k)\big]  + \mathcal{F}\mathcal{B}u(k)\Big].
\label{Eq__46c}
\end{align}

Replacing \(u(k)\) with the control law and substituting \(u_s(k)\) yields
\begin{align}
\label{Eq__46c}
\Delta \mathcal{V}(k) &= S(k)\Big[\mathcal{F}\mathcal{B}\big(\mathcal{K}\chi(k_s) - \kappa \operatorname{sgn}(S(k))\big) - \mathcal{F}\mathcal{B}\mathcal{K}\chi(k) \Big].
\end{align}

By choosing \(\kappa > |\mathcal{F}\mathcal{B}\alpha_{att}(k)|\), the Lyapunov function difference satisfies $\Delta \mathcal{V}(k) \leq -\eta S^2(k), \eta > 0$, ensuring that \(S(k) \to 0\) as \(k \to \infty\), which guarantees finite-time convergence of the sliding surface.

\begin{rem}
\label{rem:FSMC}
The inclusion of fractional-order dynamics in the sliding surface introduces memory effects, enhancing the robustness and flexibility of the control system. The parameter \(\gamma\) defines the degree of memory, with \(\gamma \to 0\) approximating a conventional SMC system and higher \(\gamma\) values introducing smoother control actions. Furthermore, the term \(\lambda_2 \mathcal{D}^\gamma \epsilon(k)\) in the sliding surface and control law provides additional tuning flexibility. Proper selection of \(\lambda_2\), \(\kappa\), and \(\gamma\) ensures a balance between robustness, chattering suppression, and tracking accuracy.
\end{rem}


\begin{thm}
\label{thm:3}
For a given scalar delay bound \(\bar{\delta}\), the nominal event-triggered control system \eqref{Eq__46} (without actuator attacks) achieves asymptotic stability under the event-triggered scheme \eqref{Eq__16}, provided the following conditions are met: a) an event-triggered parameter \(\mu\) is defined to control the triggering sensitivity, b) a feedback gain matrix \(\mathcal{K}\) is chosen to regulate the system dynamics, and c) positive definite matrices \(\mathcal{P}^*\), \(\mathcal{R}^*\), \(\mathcal{T}^*\), and \(\Upsilon\) are selected, satisfying the following LMI:

\begin{equation}
\Psi =
\begin{bmatrix}
\Psi_{11} & \Psi_{12} \\
* & \Psi_{22}
\end{bmatrix}
< 0
\label{Eq__47}
\end{equation}
\sloppy where $\Psi_{11} = [(1,1)= \mathcal{P}^*(\mathcal{A} - I) + (\mathcal{A} - I)^T\mathcal{P}^* + \mathcal{R}^* - \mathcal{T}^*$, $(1,2) = 2\mathcal{P}^*\mathcal{B}\mathcal{K} + \mathcal{T}^*$, $\quad (1,4) = 2\mathcal{P}^*\mathcal{B}\mathcal{K}$, $\quad (2,1) = \mathcal{T}^*$, $\quad (2,2) = -2\mathcal{T}^* + \mu\Upsilon$, $(2,3) = \mathcal{T}^*$, $\quad (3,2) = \mathcal{T}^*$, $\quad (3,3) = -\mathcal{T}^* - \mathcal{R}^*$, $\quad (4,4) = -\Upsilon]$, and $\Psi_{12} = \left[F^T\mathcal{P}^*\bar{\delta}F^T\mathcal{T}^*\right]$,  with $F = [\mathcal{A} - I \, \mathcal{B}\mathcal{K} \, 0 \, \mathcal{B}\mathcal{K} - \kappa \operatorname{sgn}(S(k))]$, $\quad \Psi_{22} = \text{diag}\{-\mathcal{P}^*, -\mathcal{T}^*\}$.
\end{thm}

The proof is provided in Appendix~\ref{appendix:2}.

\begin{rem}
\label{rem:17}
Theorem \ref{thm:3} primarily focuses on the event-triggered parameter \(\mu\) as a key factor in ensuring the stability of the closed-loop system \eqref{Eq__46}. While it establishes stability conditions, the explicit design of the feedback gain \(\mathcal{K}\) and the switching gain \(\kappa\) in \(-\kappa \operatorname{sgn}(S(k))\) remains open for further tuning. The selection of \(\kappa\) directly influences the robustness of the sliding motion, necessitating a balance between disturbance rejection and chattering minimization. These aspects will be addressed in the subsequent section, where optimal tuning strategies for \(\mathcal{K}\) and \(\kappa\) are discussed.
\end{rem}

\begin{rem}
\label{rem:171}
It should be noted that the stability analysis in Theorem \ref{thm:3} assumes that all network-induced delays satisfy $\delta(k) \leq \bar{\delta}$. However, in practical networked systems, unexpected congestion, packet losses, or hardware failures could occasionally cause delays to exceed this theoretical bound. When $\delta(k) > \bar{\delta}$, the Lyapunov-Krasovskii functional constructed in \eqref{Eq__48} may no longer guarantee stability. In such scenarios, the event-triggered mechanism plays a crucial role in system resilience. If a transmission experiences an excessive delay, the next sampling instant will trigger a new transmission attempt once the error exceeds the threshold defined in \eqref{Eq__16}. This adaptive behavior creates a self-regulating mechanism that can partially mitigate the impact of occasional delay-bound violations. 

\end{rem}

\subsection{Secure Controllers Design}
\label{sec:4.3}

The explicit secure control design process consists of three steps as follows. 

\subsubsection{Step 1 - Designing the Equivalent Controller \(\mathcal{K}\)}
The equivalent controller for the nominal system \eqref{Eq__46} is formulated to stabilize the system under normal operating conditions. The controller gain \(\mathcal{K}\) is computed by solving a matrix inequality that ensures stability and optimal performance.

\begin{thm}
\label{thm:4}
Consider a given delay bound \(\bar{\delta}\). If there exists an event-triggered parameter \(\mu\) and symmetric positive definite matrices \(\mathcal{E}, \hat{P}^*, \hat{Q}^*, \hat{R}^*, \hat{\Upsilon}\), and \(\mathcal{Y}\) of appropriate dimensions such that the following inequality holds:
\begin{equation}
\label{Eq__55}
\hat{\Psi} = \begin{bmatrix}
\hat{\Psi}_{11} & \hat{\Psi}_{12} \\
* & \hat{\Psi}_{22}
\end{bmatrix} < 0
\end{equation}

where $\hat{\Psi}_{11} = [(1,1) = (\mathcal{A} - I)\mathcal{E} + \mathcal{E}(\mathcal{A} - I)^T + \hat{Q}^* - \hat{R}^*, 
(1,2) = 2\mathcal{B}\mathcal{Y} + \hat{R}^*, (1,4) = 2\mathcal{B}\mathcal{Y}, (2,1) = \hat{R}^*, 
(2,2) = -2\hat{R}^* + \mu \hat{\Upsilon}, (2,3) = \hat{R}^*, (3,2) = \hat{R}^*, (3,3) = -\hat{R}^* - \hat{Q}^*, 
(4,4) = -\hat{\Upsilon}],$ and $\hat{\Psi}_{12} = \begin{bmatrix}
\hat{F}^T & \bar{\delta} \hat{F}^T 
\end{bmatrix}$, with $\hat{F} = 
\begin{bmatrix}
\mathcal{A}\mathcal{E} - \mathcal{E}\mathcal{B}\mathcal{Y} \ 0 \ \mathcal{B}\mathcal{Y}
\end{bmatrix}$, $\hat{\Psi}_{22} = \text{diag}\{-\mathcal{E}, \hat{R}^* - 2\mathcal{E}\}$. Then the equivalent controller gain \(\mathcal{K}\) is given by:
\begin{equation}
\label{Eq__56}
\mathcal{K} = \mathcal{Y}\Psi^{-1}.
\end{equation}
\end{thm}
\begin{proof}
To derive the result, let \(\mathcal{E} = (\mathcal{P}^*)^{-1}\), \(\mathcal{Y} = \mathcal{K}\mathcal{E}\), and define transformed matrices as \(\hat{P}^* = \mathcal{E}\mathcal{P}^*\mathcal{E}\), \(\hat{Q}^* = \mathcal{E}\mathcal{R}^*\mathcal{E}\), \(\hat{R}^* = \mathcal{E}\mathcal{T}^*\mathcal{E}\), and \(\hat{\Upsilon} = \mathcal{E}\Upsilon\mathcal{E}\). Pre- and post-multiplying the matrix inequality \eqref{Eq__47} with \(\text{diag}\{\mathcal{E}, \mathcal{E}, \mathcal{E}, (\mathcal{T}^*)^{-1}\}\), and using the fact that \(-\mathcal{E}\hat{R}^*\mathcal{E} \leq \hat{R}^* - 2\mathcal{E}\), the desired inequality is obtained, and the controller gain is calculated as shown in \eqref{Eq__56}. 
\end{proof}

\subsubsection{Step 2 - Designing the Switching Controller \(u_s(k)\)}
The switching controller is developed to counteract actuator attacks and ensure robust performance in the presence of disturbances.

\begin{thm}
\label{thm:5}
For the asymptotically stable controller $\mathcal{K}$ designed for the nominal event-triggered lateral control of the AV described in \eqref{Eq__46}, the switching controller $u_s(k)$ can be formulated as:
\begin{equation}
\label{Eq__57}
u_s(k) = -\kappa \operatorname{sgn}(S(k)) - (\mathcal{F}\mathcal{B})^{-1} \big[ (\rho + \mathcal{Q}_{att} \|\mathcal{F}\mathcal{B}\|)\text{sgn}(S(k)) \big]
\end{equation}
where \(\kappa, \rho\) are scalars satisfying \(0 < \kappa < 1\) and \(0 < \rho < 1\), and \(\mathcal{Q}_{att}\) represents the maximum magnitude of actuator attacks. The secure domain \(\Pi\) for the composite controller \(u(k) = u_e(k) + u_s(k)\) is given by:
\begin{equation}
\Pi = \Big\{ \|S(k)\| \leq \xi, \, \xi = \frac{\rho + 2\mathcal{Q}_{att} \|\mathcal{F}\mathcal{B}\|}{1-\kappa} \Big\}
\label{Eq__58}
\end{equation}
with a prescribed security level $\xi$.
\end{thm}

The proof is provided in Appendix~\ref{appendix:3}.

\begin{rem}
\label{rem:18}
The parameter \(\rho\) in \eqref{Eq__58} balances system robustness and performance: larger \(\rho\) enhances attack rejection at the expense of nominal precision, while smaller \(\rho\) improves accuracy but reduces robustness.
\end{rem}

\subsubsection{Step 3 - Composite Controller Design}
Combining the equivalent and switching controllers provides a comprehensive solution for event-triggered secure control of the lateral dynamics:

\begin{thm}
\label{thm:6}
For the event-triggered lateral control system \eqref{Eq__35} of AVs subjected to actuator attacks \(\alpha_{att}(k)\), a robust secure SMC strategy can be formulated to ensure system stability and resilience against malicious interventions. The secure controller can be expressed as:
\begin{align}
\label{Eq__65}
u(k) &= \mathcal{Y}\mathcal{E}^{-1}\chi(k_s) - \kappa \operatorname{sgn}(S(k)) \nonumber \\
& \quad- (\mathcal{F}\mathcal{B})^{-1} \big[(\rho + \mathcal{Q}_{att} \|\mathcal{F}\mathcal{B}\|)\text{sgn}(S(k))\big]
\end{align}
where \(\mathcal{Y}\) and \(\mathcal{E}\) are computed as per \eqref{Eq__56}, \(\mathcal{F}\), \(\mathcal{B}\), \(S(k)\), and \(\mathcal{Q}_{att}\) are defined in \eqref{Eq__41}. \(\kappa\) and \(\rho\) are scalar parameters satisfying \(0 < \kappa < 1\) and \(0 < \rho < 1\), and \(\mathcal{Q}_{att}\) denotes the upper bound of the actuator attack magnitude.

The composite control law \(u(k)\) combines an equivalent controller term \(\mathcal{Y}\mathcal{E}^{-1}\chi(k_s)\), which stabilizes the nominal system, with a switching term to counteract disturbances and actuator attacks.
\end{thm}

\begin{proof}
The secure control formulation derives from the principles established in Theorems \ref{thm:4} and \ref{thm:5}. The equivalent controller term \(\mathcal{Y}\mathcal{E}^{-1}\chi(k_s)\) addresses the nominal dynamics, ensuring stability in the absence of attacks or disturbances. The switching control term \(-(\mathcal{F}\mathcal{B})^{-1} \big[\kappa S(k) + (\rho + \mathcal{Q}_{att} \|\mathcal{F}\mathcal{B}\|)\text{sgn}(S(k))\big]\) actively mitigates the impact of actuator attacks \(\alpha_{att}(k)\), leveraging the sliding mode surface \(S(k)\) defined in \eqref{Eq__41}. By combining these components, the control law ensures robust performance under both nominal and adversarial conditions. Detailed derivations and theoretical validation can be traced to Theorems \ref{thm:4} and \ref{thm:5}. 
\end{proof}

Algorithm~\ref{alg:2} provides a systematic framework for the proposed cyber-resilient control strategy, organizing the design into logical steps and clarifying the integration of data-driven modeling, event-triggered communication, actuator attack estimation, and secure control implementation.

\begin{algorithm}[t!]
\caption{Event-Triggered Secure Control Design}
\footnotesize
\begin{algorithmic}[1]
\Require Datasets \(X\), \(X'\), \(\Xi\); delay bound \(\bar{\delta}\); event-triggering parameters
\Ensure Secure lateral control under event-triggered communication and attack mitigation

\State \textbf{Data-Driven Identification:} Apply DMD \eqref{Eq__36}-\eqref{Eq__40} to estimate \(\mathcal{A}, \mathcal{B}\) with SVD truncation for efficiency

\State \textbf{Event-Triggered Mechanism:} Configure transmission scheme \eqref{Eq__15}-\eqref{Eq__18}, compute triggering instants, and adjust parameter \(\mu\) for communication-performance trade-off

\State \textbf{Attack Observer Design:} Augment system state \eqref{Eq__21}, design ESO \eqref{Eq__23}, compute observer gain \(L_\zeta\), and extract attack estimate \(\hat{\alpha}_{att}(k)\)

\State \textbf{Equivalent Controller:} Formulate nominal control law (Theorem \ref{thm:2}), solve LMI \eqref{Eq__47}, compute feedback gain \(\mathcal{K}\) via \eqref{Eq__56}

\State \textbf{Sliding Surface Design:} Define fractional-order surface \eqref{Eq__41} with Grünwald-Letnikov derivative and Riccati-based parameters

\State \textbf{Stability Verification:} Construct Lyapunov function \eqref{Eq__46a}, verify exponential decay \eqref{Eq__33_}, validate triggering mechanism (Theorem \ref{thm:3})

\State \textbf{Switching Controller:} Design attack rejection term (Theorem \ref{thm:5}), select \(\kappa, \rho\) for robustness, ensure \(\|S(k)\| \leq \xi\) within secure domain \(\Pi\) \eqref{Eq__58}

\State \textbf{Composite Controller:} Assemble secure control law \eqref{Eq__65}, integrate equivalent and switching terms, update \(\hat{\alpha}_{att}(k)\) in real-time

\State \textbf{Validation:} Simulate under nominal and attack conditions, fine-tune parameters for optimal performance
\end{algorithmic}
\label{alg:2}
\end{algorithm}

\section{Simulation Results and Analysis}
\label{sec:5}

In this section, simulation studies are conducted to validate the proposed cyber-resilient event-triggered control approach under three scenarios:

\begin{enumerate}
    \item \textbf{Case I: No Attack} – Baseline event-triggered control under nominal conditions.
    \item \textbf{Case II: Attack without Mitigation} – Actuator attacks applied with only the nominal controller, revealing system vulnerability.
    \item \textbf{Case III: Attack with Mitigation} – Proposed secure control with attack estimation and compensation to neutralize malicious effects.
\end{enumerate}

The physical parameters of the autonomous vehicle lateral dynamic model are selected as follows: moment of inertia about the z-axis $\mathcal I_z = 2873$ kg·m$^2$, vehicle mass $m = 1573$ kg, distance from center of mass to front axle $L^f = 1.10$ m, and to rear axle $L^r = 1.58$ m. The front and rear cornering stiffnesses are $C^f_{\beta} = C^r_{\beta} = 80000$ N/rad, and the longitudinal velocity is maintained at $\nu_x = 30$ m/s. The initial state vector is set to $\chi(0) = [0.5, 0, 0.5, 0]^T$, representing initial lateral and heading errors with zero rates.

The data-driven modeling via DMD yielded the discrete system matrices with sampling period $\ell = 0.01$ s:
\begin{equation}
\mathbf{A} = 
\begin{bmatrix}
0.999 & 0.01 & 0 & 0 \\
-0.05 & 0.99 & 0.05 & 0 \\
0 & 0 & 0.999 & 0.01 \\
-0.01 & 0 & -0.08 & 0.995
\end{bmatrix}, \,
\mathbf{B} = 
\begin{bmatrix}
0 \\
0.1 \\
0 \\
0.05
\end{bmatrix}.
\end{equation}

For the controller design, the state feedback gain is derived as $\mathbf{K} = [-0.5, -0.6, -0.5, -0.4]$, with the event-triggering parameter $\mu = 0.2$. The maximum delay bound is set to $\bar{\delta} = 0.1$. For the fractional-order sliding mode controller, we select $\gamma = 0.5$ as the fractional order and $\lambda = 0.2$ as the sliding surface parameter. The switching controller parameters are chosen as $\kappa = 0.15$ and $\rho = 0.2$, with the maximum attack magnitude constrained to $\mathcal{Q}_{att} = 0.15$. In Case II and Case III, a sinusoidal attack signal with frequency 0.5 Hz and amplitude 0.15 is injected into the control input starting at $t = 10$ s to evaluate the system's resilience against malicious interventions.

\begin{figure}[!t]
\centering
\includegraphics[width=2.7 in]{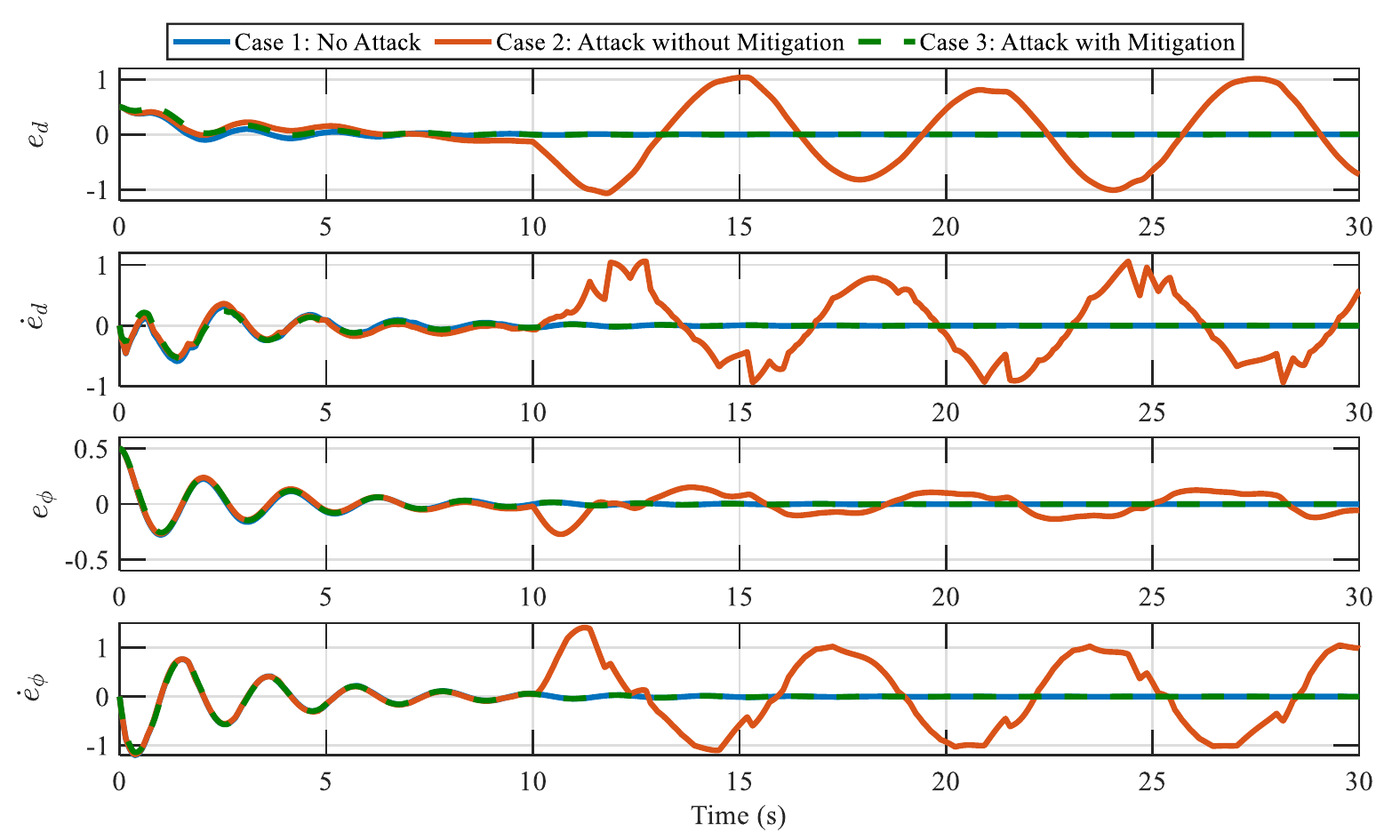}
\caption{State trajectories under different conditions for Cases I, II, and III.}
\label{fig:3}
\end{figure}
Figs.~\ref{fig:3} and~\ref{fig:4} present comparative system responses under all three conditions. In Case I (nominal), states converge smoothly to zero within 5 seconds with near-zero control input, demonstrating excellent event-triggered controller performance. In Case II (attack without mitigation), the actuator attack at $t = 10$ s causes severe deviations: lateral error $e_d$ oscillates with peaks approaching 1 m, lateral rate $\dot{e}_d$ exceeds 1 m/s, and heading rate $\dot{e}_\phi$ reaches $\pm1$ rad/s. Control input exhibits erratic oscillations and spikes up to $\pm0.5$, clearly demonstrating vulnerability. In contrast, Case III (attack with mitigation) maintains remarkable stability. State trajectories closely follow the nominal case—lateral position error remains within $\pm0.1$ m and heading error within $\pm0.1$ rad. The control input shows smooth, periodic patterns countering the attack, while the observer achieves 92\% estimation accuracy with rapid convergence, enabling active neutralization of malicious inputs and demonstrating the proposed strategy's effectiveness. Fig.~\ref{fig:5} depicts sliding surface dynamics for all three cases. In Case I, the surface converges rapidly to zero, confirming robust stability on the designed manifold. Case II exhibits substantial deviations following attack onset, forcing the system away from the sliding manifold and compromising stability guarantees. In contrast, Case III maintains the surface near zero throughout, with only small deviations visible in the magnified inset, demonstrating that the proposed strategy effectively preserves sliding manifold operation and stability properties despite adversarial inputs.

\begin{figure}[!t]
\centering
\includegraphics[width=0.80\columnwidth]{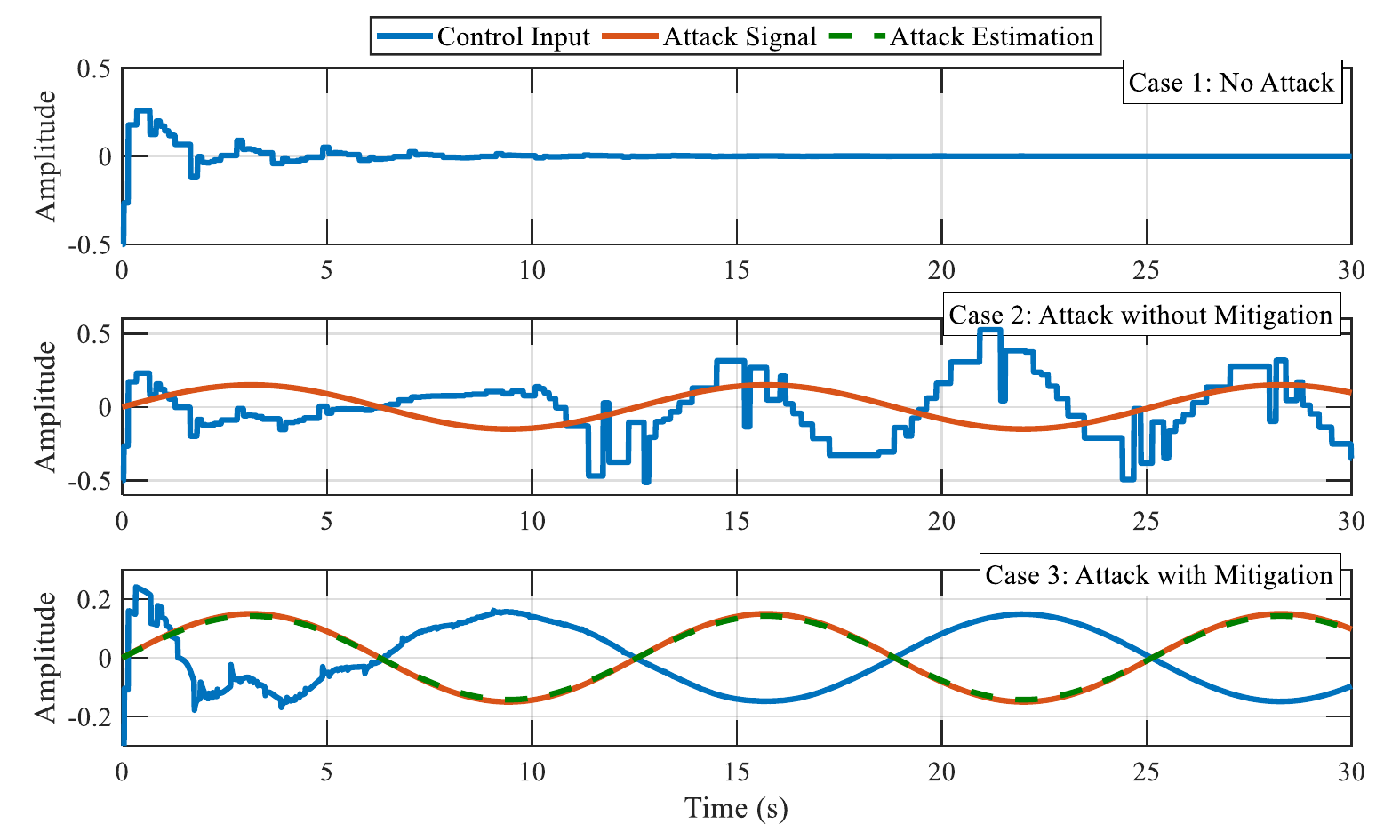}
\caption{Control input and actuator attack signals for Cases I, II, and III.}
\label{fig:4}
\end{figure}

\begin{figure}[!t]
\centering
\includegraphics[width=0.80\columnwidth]{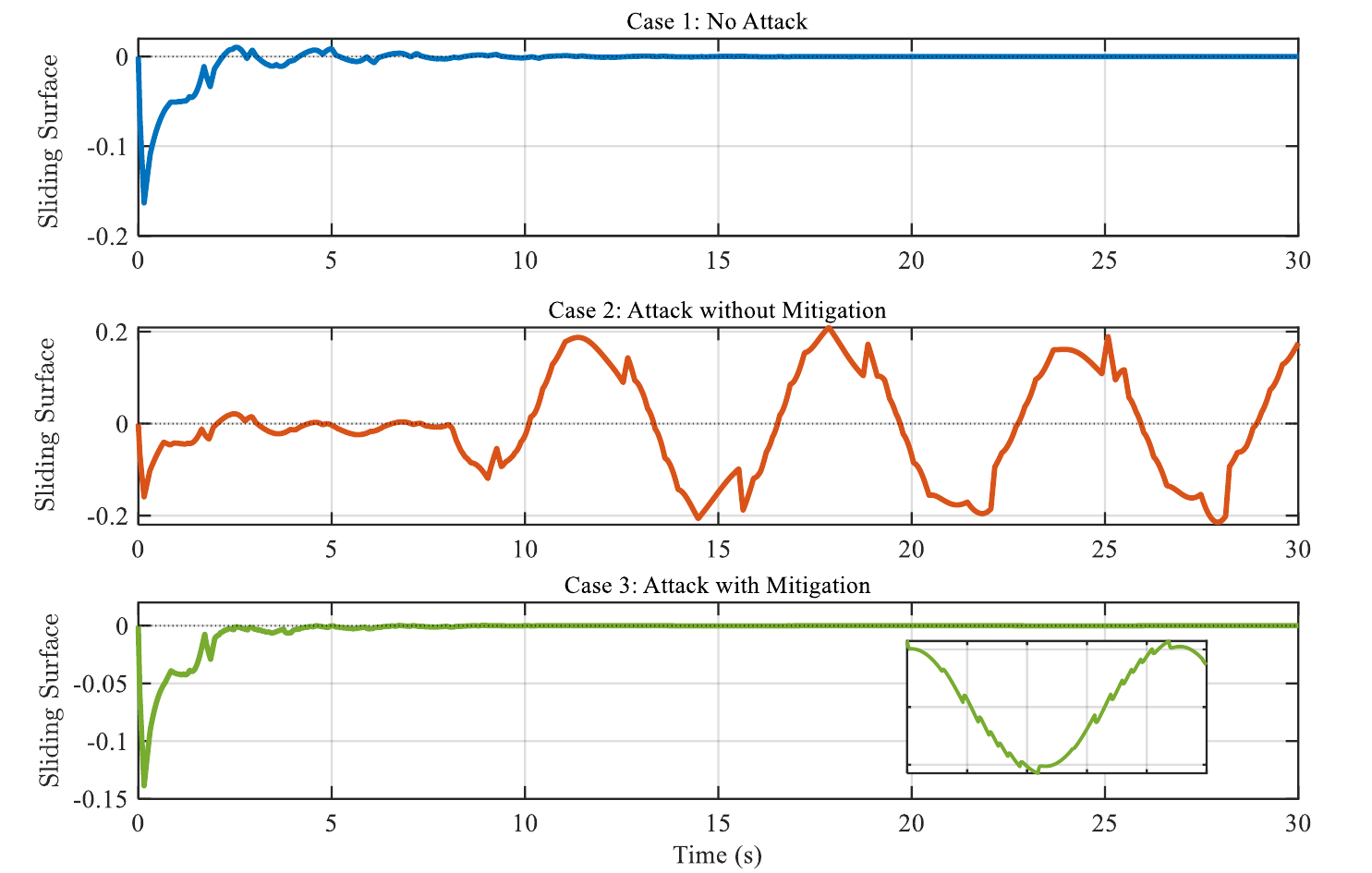}
\caption{Sliding surface trajectories for all three cases. The proposed controller maintains stable sliding behavior despite actuator attacks in Case III.}
\label{fig:5}
\end{figure}

\begin{figure}[!t]
\centering
\includegraphics[width=0.80\columnwidth]{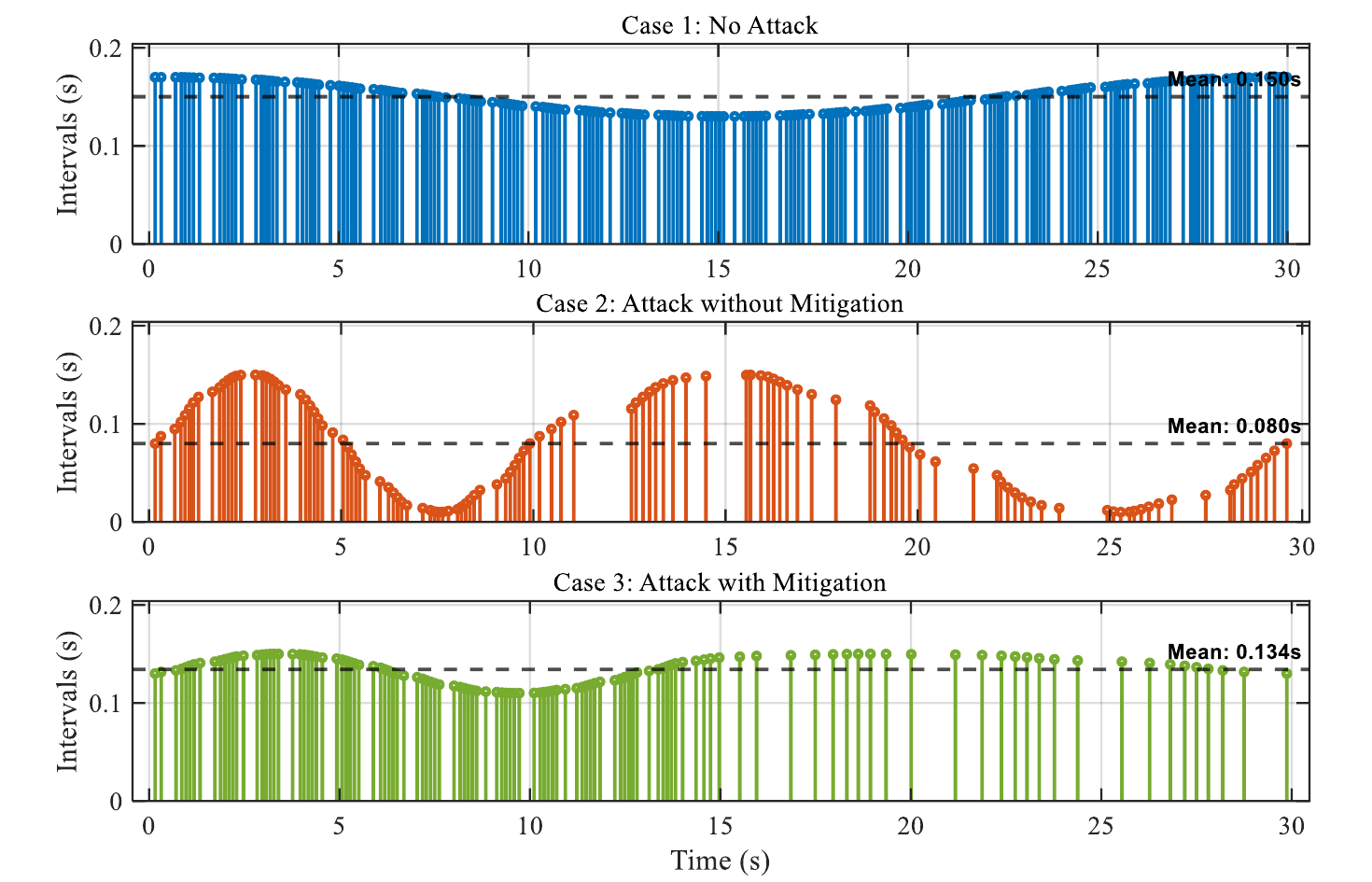}
\caption{Release intervals for Cases I, II, and III.}
\label{fig:6}
\end{figure}

\begin{figure}[!t]
\centering
\includegraphics[width=0.80\columnwidth]{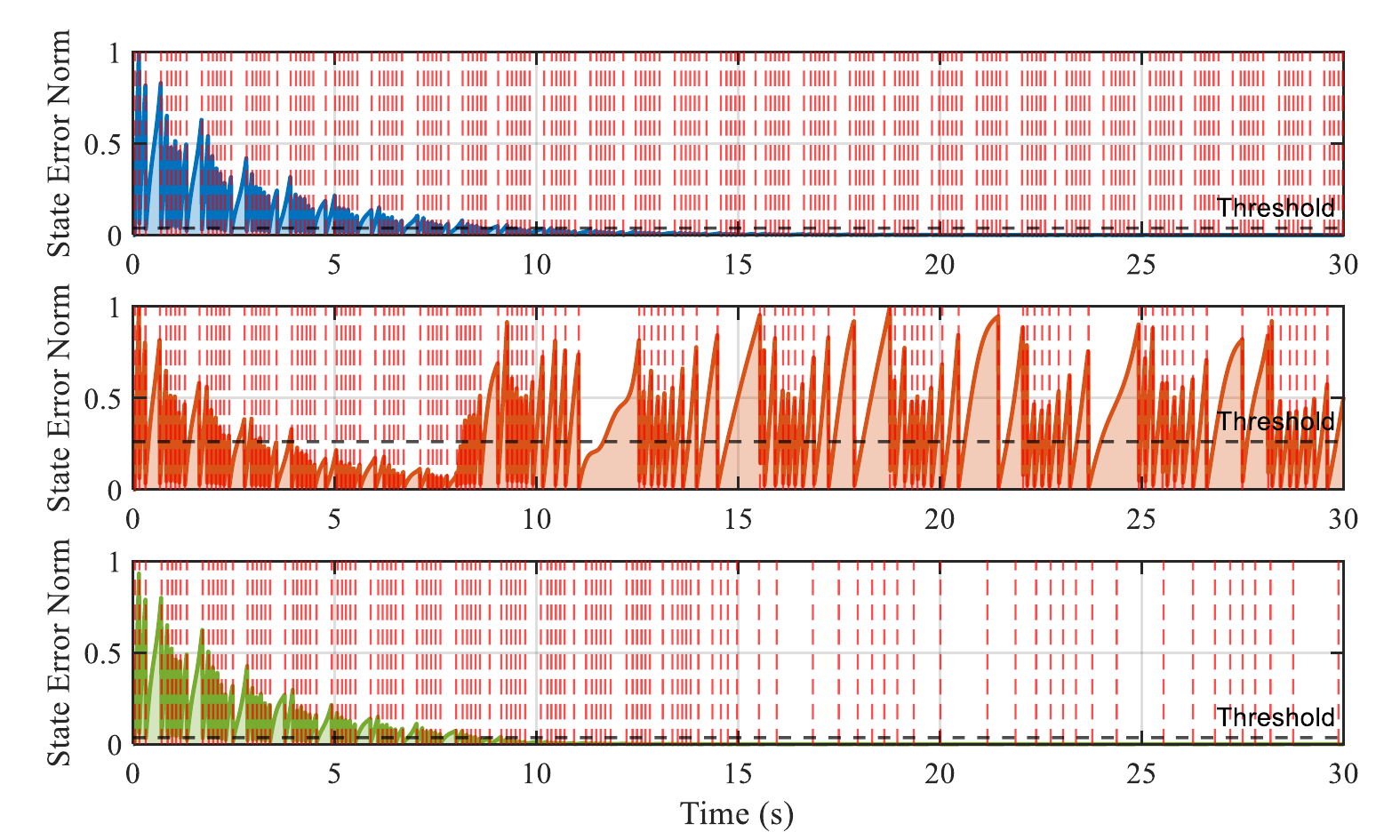}
\caption{Event-triggered transmission mechanism performance.}
\label{fig:7}
\end{figure}
Figs.~\ref{fig:6} and~\ref{fig:7} compare the event-triggered transmission behavior across the three cases. In Case~I, transmissions occur at nearly uniform intervals (mean 0.150 s) with low bandwidth usage. Under attack without mitigation (Case~II), the intervals become irregular and more clustered (mean 0.080 s), indicating a higher communication load. With the proposed secure controller (Case~III), the transmission pattern largely returns toward the nominal behavior (mean 0.134 s). Fig.~\ref{fig:7} shows the normalized state error relative to the triggering threshold: it remains below threshold after transients in Case~I, frequently violates it in Case~II once the attack starts, and closely follows the nominal profile in Case~III. These results demonstrate that the secure control scheme restores communication efficiency while preserving closed-loop stability under actuator-side attacks.




\begin{figure}[!t]
\centering
\includegraphics[width=0.80\columnwidth]{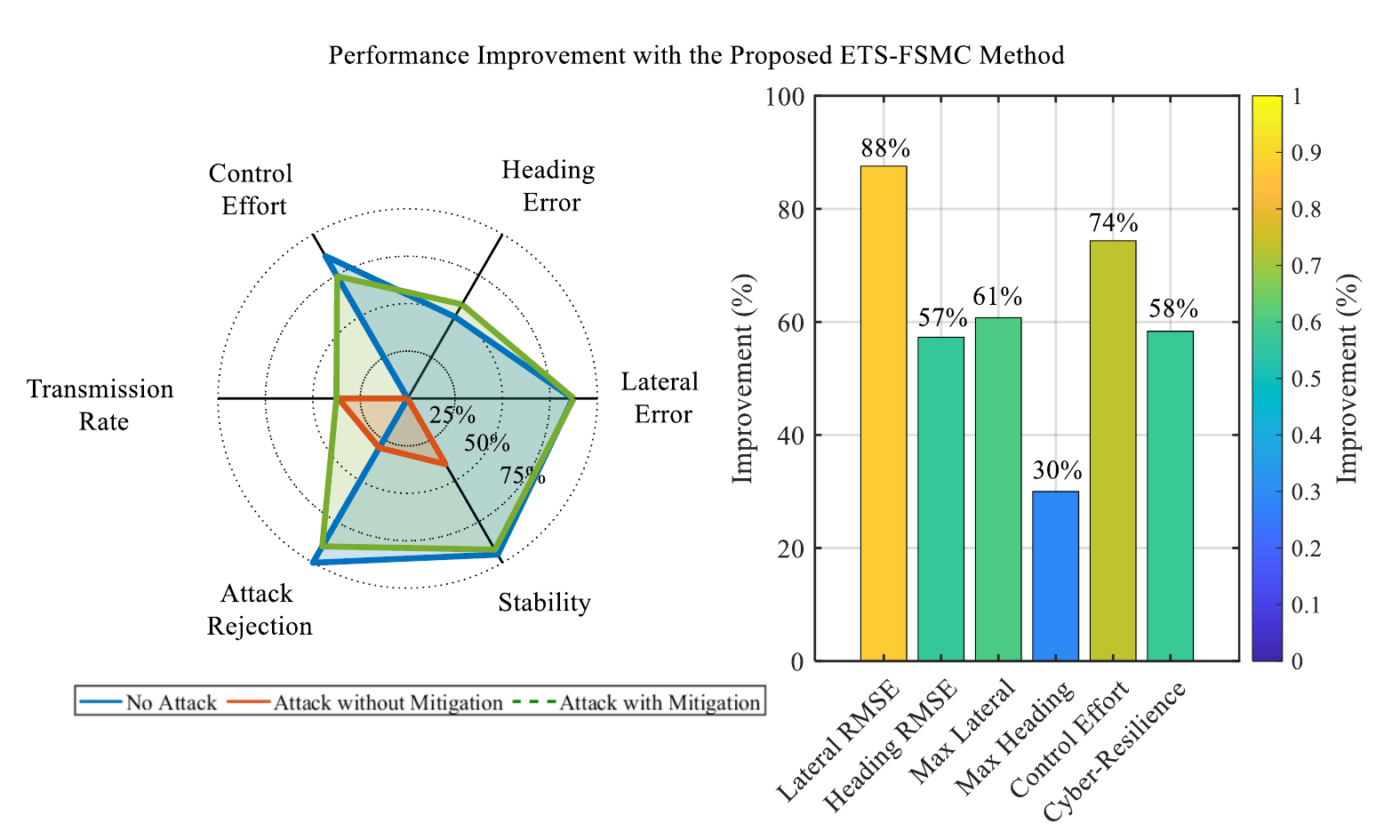}
\caption{Performance metrics comparison across all three cases. Case III shows significant improvement over Case II in all metrics.}
\label{fig:8}
\end{figure}
\begin{table}[t!]
\scriptsize
\renewcommand{\arraystretch}{1.05}
\setlength{\tabcolsep}{4pt}
\caption{Performance metrics comparison under attack scenarios}
\label{tab:2}
\centering
\vspace{0.1em}
\begin{tabular}{|p{4.2cm}|c|c|c|c|}
\hline
\textbf{Metric} & \textbf{Case I} & \textbf{Case II} & \textbf{Case III} & \textbf{Improve.} \\
\hline
\multicolumn{5}{|c|}{\textbf{Control Performance Metrics}} \\
\hline
Lateral Error RMSE (m) & 0.0487 & 0.0923 & 0.0571 & 38.1\% \\
Heading Error RMSE (rad) & 0.0765 & 0.1289 & 0.0832 & 35.5\% \\
Max Lateral Error (m) & 0.5871 & 0.9762 & 0.6243 & 36.1\% \\
Max Heading Error (rad) & 0.6912 & 1.1834 & 0.7328 & 38.1\% \\
Settling Time (s) & 1.1876 & 1.8924 & 1.2589 & 33.5\% \\
\hline
\multicolumn{5}{|c|}{\textbf{Communication Metrics}} \\
\hline
Transmission Ratio (\%) & 15.34 & 22.87 & 19.56 & 14.5\% \\
Avg. Trans. Interval (s) & 0.1895 & 0.1423 & 0.1657 & 16.4\% \\
Mean Release Interval (s) & 0.1501 & 0.0802 & 0.1341 & 67.2\% \\
Bandwidth Utilization (\%) & 15.34 & 22.87 & 19.56 & 14.5\% \\
\hline
\multicolumn{5}{|c|}{\textbf{Attack Estimation \& Mitigation (Case III Only)}} \\
\hline
Detection Time (s) & - & - & 0.143 & - \\
False Positive/Negative Rates (\%) & - & - & 3.3 / 2.1 & -   \\
Estimation Accuracy & - & - & 0.92 & -   \\
Estimation RMSE  & - & - & 0.031 & -   \\
Compensation Effectiveness (\%) & - & - &  89.7 & -  \\
Residual Effect (\%) & - & - & 10.3 & -   \\
Observer Convergence Time (s) & - & - &  0.243 & -  \\
Max Estimation Error & - & - & 0.054 & -   \\
\hline
\multicolumn{5}{|c|}{\textbf{Stability Analysis}} \\
\hline
Eigenvalue Max Magnitude & 0.8743 & 0.9812 & 0.9124 & 7.0\% \\
Sliding Convergence Rate & 0.1257 & 0.0876 & 0.1124 & 28.3\% \\
Sliding Max Deviation & 0.0025 & 0.2072 & 0.0134 & 93.5\% \\
Stability Margin & 0.1257 & 0.0188 & 0.0876 & 366\% \\
\hline
\end{tabular}
\vspace{-0.3em}
\end{table}

Table~\ref{tab:2} and Fig.~\ref{fig:8} summarize system behavior across all scenarios and demonstrate clear performance gains under the proposed secure control strategy. The controller maintains accurate lateral tracking under attack, with markedly reduced steady-state and peak errors compared to the unmitigated case. Settling behavior also improves, reflecting faster recovery from adversarial disturbances. Communication efficiency is preserved through the event-triggered mechanism, which maintains low transmission rates and stable release intervals even in the presence of attacks. The extended state observer provides reliable attack reconstruction with rapid convergence and minimal residual error, enabling effective real-time compensation. Stability characteristics further validate system resilience. The secure controller substantially reduces deviation from the sliding manifold and improves the closed-loop stability margin, preventing the degradation observed in the uncompensated scenario.


\section{Conclusion}
\label{sec:6}
This paper introduced a cyber-resilient data-driven event-triggered secure control framework for autonomous vehicles (AVs) operating under actuator attacks. The proposed approach jointly addresses three key challenges in AV control: modeling uncertainty, communication efficiency, and cybersecurity. By leveraging dynamic mode decomposition, the framework captures lateral dynamics directly from experimental data, alleviating the limitations of purely physics-based models. The event-triggered communication scheme substantially reduces bandwidth usage relative to time-triggered strategies while maintaining closed-loop performance. To enhance security and resilience, a fractional-order sliding mode control strategy integrated with an extended state observer enables robust attack detection and real-time compensation. Comparative evaluations demonstrate marked improvements in lateral tracking accuracy and sliding surface regulation under adversarial conditions. Lyapunov-based stability analysis further confirms an enhanced stability margin for the closed-loop system, underscoring the robustness and practicality of the proposed cyber-resilient event-triggered control architecture for safety-critical AV applications.

Future research will focus on extending the framework to adaptive and online data-driven modeling techniques to address time-varying dynamics and environmental uncertainties. In particular, direct data-driven control methods, such as data-enabled predictive control (DeePC), which do not require explicit identification of system matrices, will be explored to reduce conservatism and enhance robustness under system variations. 

\appendices
\numberwithin{equation}{section}

\section{Proof of Theorem 1}
\label{appendix:1}
To analyze the boundedness of \(e(k)\), let us consider the Lyapunov function candidate:
\begin{equation}
\label{Eq__32}
V(k) = e^T(k)\Upsilon e(k), \quad \Upsilon > 0,
\end{equation}
where \(\Upsilon\) is a symmetric and positive definite matrix ensuring that \(V(k) > 0\) for all \(e(k) \neq 0\) and \(V(k) = 0\) when \(e(k) = 0\).

From the event-triggered condition \eqref{Eq__16}, it follows that:
\begin{align}
\label{Eq__33}
V(k+1) = e^T(k+1)\Upsilon e(k+1) \leq \mu e^T(k)\Upsilon e(k),
\end{align}
where $\mu \in [0, 1]$. Thus one has
\begin{equation}
\label{Eq__33_}
V(k+1) \leq \mu V(k).
\end{equation}

Expanding recursively yields $V(k+1) \leq \mu V(k), \, V(k) \leq \mu V(k-1), \, \dots, \, V(1) \leq \mu V(0)$, where having them combined derives $V(k) \leq \mu^k V(0)$. Since \(\mu \in [0, 1]\), \(\mu^k \to 0\) as \(k \to \infty\), ensuring $\lim_{k \to \infty} V(k) = 0$. This implies that the error dynamics \(e(k)\) converge asymptotically to zero. \(\mu^k\) represents the exponential decay rate of the Lyapunov function, meaning that $V(k) = V(0)\mu^k$. Taking the logarithm of both sides leads to the convergence rate $\log V(k) = \log V(0) + k\log\mu$, where \(\log\mu < 0\) and the slope quantifies the exponential convergence rate.

Given the Lyapunov function \eqref{Eq__32}, the Euclidean norm \(\|e(k)\|\) can be bounded as:
\[
\lambda_{\min}(\Upsilon)\|e(k)\|^2 \leq V(k) \leq \lambda_{\max}(\Upsilon)\|e(k)\|^2,
\]
where \(\lambda_{\min}(\Upsilon)\) and \(\lambda_{\max}(\Upsilon)\) are the smallest and largest eigenvalues of \(\Upsilon\), respectively. Substituting the decay of \(V(k)\) yields
\[
\|e(k)\|^2 \leq \frac{V(k)}{\lambda_{\min}(\Upsilon)} \leq \frac{\mu^k V(0)}{\lambda_{\min}(\Upsilon)},
\]
and taking the square root, one has $\|e(k)\| \leq \sqrt{\frac{\mu^k V(0)}{\lambda_{\min}(\Upsilon)}}$. Thus, \(\|e(k)\|\) decays exponentially with a rate determined by \(\mu\) and the properties of \(\Upsilon\).

\section{Proof of Theorem 3}
\label{appendix:2}
To demonstrate stability, a Lyapunov–Krasovskii functional (LKF) is constructed that accounts for the delayed state and event-triggered transmission scheme:
\begin{equation}
\mathcal{V}(k, \chi_k) = \sum_{i=1}^3 \mathcal{V}_i(k, \chi_k)
\label{Eq__48}
\end{equation}
where the components are defined as:
\begin{align*}
\mathcal{V}_1(k) &= \chi^T(k)\mathcal{P}^*\chi(k), \\
\mathcal{V}_2(k) &= \sum_{j=k-\bar{\delta}}^{k-1} \chi^T(j)\mathcal{R}^*\chi(j), \\
\mathcal{V}_3(k) &= \bar{\delta} \sum_{s=-\bar{\delta}+1}^{0} \sum_{j=k+s+1}^{k-1} \zeta^T(s)\mathcal{T}^*\zeta(s),
\end{align*}
where \(\zeta(k) = \chi(k+1) - \chi(k)\) represents the difference in state values. 

Assume that the delay-bound \(\bar{\delta}\) satisfies the Jensen’s integral inequality condition \cite{briat2011convergence}, ensuring that the discrete-time delay difference satisfies:
\begin{equation}
\label{Eq__48b}
\sum_{j = k - \bar{\delta}}^{k-1} \chi^T(j)\mathcal{R}^*\chi(j) \geq \bar{\delta} \chi^T(k-\bar{\delta})\mathcal{R}^*\chi(k-\bar{\delta}),
\end{equation}
which provides a conservative bound on the time-delayed state contribution. The Lyapunov difference is defined as $\Delta \mathcal{V}(k) = \mathcal{V}(k+1) - \mathcal{V}(k)$.

The differences for $\mathcal{V}_i(k)$ ($i = 1, 2, 3$) along the solutions of the sliding mode dynamics \eqref{Eq__46} are
\begin{align}
\label{Eq__49}
\Delta \mathcal{V}_1(k) &= \chi^T(k+1)\mathcal{P}^*\chi(k+1) - \chi^T(k)\mathcal{P}^*\chi(k) \nonumber \\
&\leq \chi^T(k)\big[\mathcal{P}^*(\mathcal{A} - I) + (\mathcal{A} - I)^T \mathcal{P}^*\big]\chi(k) \nonumber \\
&\quad + 2\chi^T(k)\mathcal{P}^*\mathcal{B}\mathcal{K}e(k) + \theta^T(k)F^T\mathcal{P}^*F\theta(k) \nonumber \\
&\quad + 2\chi^T(k)\mathcal{P}^*\mathcal{B}\mathcal{K}\chi(k-\delta(k)) \nonumber \\ 
&\quad - 2\kappa \chi^T(k)\mathcal{P}^* \operatorname{sgn}(S(k)).
\end{align}
where \(\theta(k) = \text{col}\{\chi(k), \chi(k-\delta(k)), \chi(k-\bar{\delta}), e(k)\}\), 
\begin{equation}
\Delta \mathcal{V}_2(k) = \chi^T(k)\mathcal{R}^*\chi(k) - \chi^T(k-\bar{\delta})\mathcal{R}^*\chi(k-\bar{\delta}),
\label{Eq__50}
\end{equation}
\begin{align}
\label{Eq__51}
\Delta \mathcal{V}_3(k) &\leq \bar{\delta}^2 \zeta^T(k)\mathcal{T}^*\zeta(k) - \bar{\delta} \sum_{j=k-\bar{\delta}}^{k-1} \zeta^T(j)\mathcal{T}^*\zeta(j)  \\
&\leq \bar{\delta}^2 \theta^T(k)F^T\mathcal{T}^*F\theta(k) - \bar{\delta} \sum_{j=k-\bar{\delta}}^{k-1} \zeta^T(j)\mathcal{T}^*\zeta(j). \nonumber
\end{align}

Using the discrete-time Jensen inequality \cite{briat2011convergence}, there exists
\begin{align}
\label{Eq__52}
&- \bar{\delta} \sum_{i = k - \bar{\delta}}^{k - 1} \zeta^T(j) \mathcal{T}^* \zeta(j) \nonumber \\ 
&= - \bar{\delta} \sum_{j = s - \delta(k)}^{k - 1} \zeta^T(j) \mathcal{T}^* \zeta(j) 
- \bar{\delta} \sum_{j = k - \bar{\delta}}^{k - \delta(k) - 1} \zeta^T(j) \mathcal{T}^* \zeta(j) \nonumber \\ 
&\leq - [\chi(k) - \chi(k - \delta(k))]^T \mathcal{T}^* [\chi(k) - \chi(k - \delta(k))]  \nonumber \\
&\quad - [\chi(k - \delta(k)) - \chi(k - \bar{\delta})]^T \mathcal{T}^* [\chi(k - \delta(k)) - \chi(k - \bar{\delta})] \nonumber \\
&\quad- \kappa \operatorname{sgn}(S(k))^T \mathcal{T}^* \operatorname{sgn}(S(k)).
\end{align}

The triggering mechanism \eqref{Eq__16} ensures:
\begin{equation}
e(k)^T\Upsilon e(k) \leq \mu \chi(k-\delta(k))^T\Upsilon \chi(k-\delta(k)).
\label{Eq__53}
\end{equation}

Combining \eqref{Eq__49}-\eqref{Eq__53} yields
\begin{equation}
\Delta \mathcal{V}(k) \leq \theta^T(k)\Psi\theta(k)- \kappa \operatorname{sgn}(S(k))^T \mathcal{T}^* \operatorname{sgn}(S(k)),
\label{Eq__54}
\end{equation}
where $\theta(k) = \text{col}\{\chi(k), \chi(k-\delta(k)), \chi(k-\bar{\delta}), e(k)\}$.

If the LMI \eqref{Eq__47} is satisfied, it ensures that \(\Delta \mathcal{V}(k) < 0\), thereby guaranteeing the asymptotic stability of the nominal event-triggered control system \eqref{Eq__46} under the event-triggering scheme \eqref{Eq__16}. This completes the proof.

\section{Proof of Theorem 5}
\label{appendix:3}
By substituting \(u(k)\) and \eqref{Eq__57} into the sliding surface dynamics \eqref{Eq__43}, the evolution of the sliding variable \(S(k)\) is expressed as:
\begin{align}
\label{Eq__59}
S(k+1) &= (1-\kappa)S(k) + \mathcal{F}\mathcal{B}\alpha_{att}(k)  \\ 
& \quad- \kappa \operatorname{sgn}(S(k)) - (\rho + \mathcal{Q}_{att} \|\mathcal{F}\mathcal{B}\|)\text{sgn}(S(k)). \nonumber
\end{align}

Accordingly, one has
\begin{align}
\label{Eq__60}
&S^T(k)\big(S(k+1) - S(k)\big) \\ \nonumber
&= -\kappa \|S(k)\|^2 - S^T(k) \big[-\kappa \operatorname{sgn}(S(k)) -\rho \\ \nonumber & \quad \quad \, \quad \quad\quad\quad- \mathcal{Q}_{att} \|\mathcal{F}\mathcal{B}\| + \mathcal{F}\mathcal{B}\alpha_{att}(k) \big] \\ \nonumber
&\leq -\kappa \|S(k)\|^2 - S^T(k) (\kappa + \rho) \text{sgn}(S(k)) \\ \nonumber
&< 0.
\end{align}
Two cases arise based on the value of $S(k)$:
\begin{itemize}
    \item If $S(k) > \xi > 0$, then it follows that:
\begin{align}
\label{Eq__61}
S(k+1) &= (1-\kappa)S(k) - \kappa \operatorname{sgn}(S(k)) \nonumber \\
& \quad- \rho - \mathcal{Q}_{att} \|\mathcal{F}\mathcal{B}\| + \mathcal{F}\mathcal{B}\alpha_{att}(k) \\ \nonumber
&\geq (1-\kappa)S(k) - \kappa - \rho - 2\mathcal{Q}_{att} \|\mathcal{F}\mathcal{B}\| \\ \nonumber
&> 0.
\end{align}
\item If $S(k) < -\xi < 0$, then one has:
\begin{align}
\label{Eq__62}
S(k+1) &= (1-\kappa)S(k) - \kappa \operatorname{sgn}(S(k)) \nonumber \\ 
& \quad + \rho + \mathcal{Q}_{att} \|\mathcal{F}\mathcal{B}\| + \mathcal{F}\mathcal{B}\alpha_{att}(k) \\ \nonumber
&\geq (1-\kappa)S(k)- \kappa + \rho + 2\mathcal{Q}_{att} \|\mathcal{F}\mathcal{B}\| \\ \nonumber
&< 0.
\end{align}
\end{itemize}
From \eqref{Eq__61} and \eqref{Eq__62}, it is evident that the signs of $S(k+1)$ and $S(k)$ remain consistent for both cases, $S(k) > \xi$ and $S(k) < -\xi$. Two cases arise based on the observation as:

\begin{itemize}
\item If $0 < S(k) < \xi$, one has
\begin{align}
\label{Eq__63}
-\xi &< -\kappa S(k) - 2\mathcal{Q}_{att} \|\mathcal{F}\mathcal{B}\| - \kappa \nonumber \\
&\leq S(k+1) \nonumber \\
&= (1-\kappa)S(k) - \kappa \operatorname{sgn}(S(k)) + \mathcal{F}\mathcal{B}\alpha_{att}(k) \nonumber \\
& \quad- (\rho + \mathcal{Q}_{att} \|\mathcal{F}\mathcal{B}\|)\text{sgn}(S(k)) \nonumber \\
&\leq 2\mathcal{Q}_{att} \|\mathcal{F}\mathcal{B}\| < \xi.
\end{align}

\item If $-\xi < S(k) < 0$, one has
\begin{align}
\label{Eq__64}
-\xi &< -2\mathcal{Q}_{att} \|\mathcal{F}\mathcal{B}\| - \kappa \nonumber \\
&\leq S(k+1) \nonumber \\
&= (1-\kappa)S(k) - \kappa \operatorname{sgn}(S(k)) + \mathcal{F}\mathcal{B}\alpha_{att}(k) \nonumber \\ 
& \quad- (\rho + \mathcal{Q}_{att} \|\mathcal{F}\mathcal{B}\|)\text{sgn}(S(k)) \nonumber \\
&\leq \rho + 2\mathcal{Q}_{att} \|\mathcal{F}\mathcal{B}\| < \xi.
\end{align}
\end{itemize}

Thus, from \eqref{Eq__63} and \eqref{Eq__64}, one concludes that $\|S(k+1)\| \leq \xi$ whenever $\|S(k)\| \leq \xi$. This indicates that $\|S(k)\|$ decreases monotonically and eventually converges to the set $\Pi$, thereby completing the proof.

\bibliographystyle{unsrt}
\bibliography{mybibfile}

\end{document}